\documentclass[9pt,letter,sigconf]{acmart}

\usepackage{xcolor,colortbl}
\usepackage{booktabs} % For formal tables

\usepackage{comment}
\usepackage{psfrag}
\usepackage{amsfonts}
\usepackage{amsmath}
\usepackage{amssymb}
\usepackage{amsthm}
\usepackage{mathtools}
\usepackage{alltt}
\usepackage{longtable}
\usepackage{subcaption}
\usepackage{epsfig}
\usepackage{epstopdf}
\usepackage{fancyvrb}

\newtheoremstyle{mytheoremstyle}{2pt}{1pt}{\itshape}{}{\bfseries}{.}{.5em}{} 
\theoremstyle{mytheoremstyle}
\usepackage{xpatch}
\makeatletter
\xpatchcmd{\proof}{\topsep6\p@\@plus6\p@\relax}{}{}{}
\makeatother

\usepackage[T1]{fontenc}
\makeatletter
\DeclareFontFamily{T1}{lcmtt}{\hyphenchar\font\m@ne}
\makeatother

\DeclareFontShape{T1}{lcmtt}{m}{n}{%
	<13.82><16.59><9><23.89><28.66><34.4><41.28>%
	ecltt8}{}

\usepackage{enumerate}
\usepackage[ruled]{algorithm2e} % For algorithms

\SetAlFnt{\small}
\SetAlCapFnt{\small}
\SetAlCapNameFnt{\small}
\SetAlCapHSkip{0pt}
\IncMargin{-\parindent}

\usepackage{mathtools}
\usepackage{enumitem}

\setlist[description]{leftmargin=\parindent,labelindent=\parindent}

\usepackage{epstopdf}
\usepackage{etoolbox}
\makeatletter
\preto{\@verbatim}{\topsep=0pt \partopsep=0pt }
\makeatother
\newtheorem{definition}{Definition}

\newtheorem{theorem}{Theorem}
\newtheorem{example}{Example}

\newenvironment{mycenter}[1][\topsep]
{\setlength{\topsep}{#1}\par\kern\topsep\centering}% \begin{mycenter}[<len>]
{\par\kern\topsep}% \end{mycenter}

\renewcommand\footnotetextcopyrightpermission[1]{} % removes footnote with conference information in first column
\begin{document}
	%\null
	%\includepdf[pages=1]{cover.pdf}
	
		\author{Antonio A. Bruto da Costa}
		\orcid{0000-0002-4590-0665}
		\affiliation{%
			\institution{Indian Institute of Technology}
			\streetaddress{}
			\city{Kharagpur}
			\state{West Bengal}
			\postcode{721302}
			\country{India}}
		\email{antonio.cse.iitkgp@gmail.com}
		\author{Pallab Dasgupta}
		\orcid{0000-0002-4590-0665}
		\affiliation{%
			\institution{Indian Institute of Technology}
			\streetaddress{}
			\city{Kharagpur}
			\state{West Bengal}
			\postcode{721302}
			\country{India}}
		\email{pallab@cse.iitkgp.ac.in}
		\author{Goran Frehse}
		\affiliation{%
			\institution{VERIMAG Laboratory}
			\streetaddress{}
			\city{Grenoble}
			\state{}
			\country{France}}
		\email{goranf@gmail.com}
	\title[Formal Feature Interpretation of Hybrid Systems]{Formal Feature Interpretation of Hybrid Systems}
	
	\begin{abstract}
		In current practice a {\em formal} analysis of hybrid system models is assertion-based. The work presented here is based on {\em features} that look beyond functional correctness towards a quantitative evaluation of behavioural attributes. A feature defines a real-valued evaluation function over a specific set of traces. This article describes an improved method for the interpretation of features over hybrid automata models. It further demonstrates how Satisfiability Modulo Theory (SMT) solvers can be used for extracting behavioural traces corresponding to corner cases of a feature. Results are demonstrated on examples from the control and circuit domains.
		
		%This article describes an improved methodology for the interpretation of {\em features} over hybrid automata models. It further demonstrates how Satisfiability Modulo Theory (SMT) solvers can be used for extracting behavioural traces corresponding to corner cases of a feature. Unlike past work, where a {\em formal} analysis of hybrid system models is assertion-based, the work presented here is rooted in the use of features that look beyond functional correctness to evaluate the margins of behavioural attributes. Results are demonstrated on examples taken from  the control and circuit domains. The use of SMT solvers adds power to the formal analysis framework by providing a trace of the behavioural corners for a feature. This is shown on an atomic reactor temperature control strategy.		
	\end{abstract}

	
	\maketitle
	
	\renewcommand{\shortauthors}{XXXX et al.}
	
	\section{Introduction}\label{sec:introduction}
	
	The theory of {Hybrid Automata (HA)} has been extensively studied in the context of designing provably safe designs of embedded hybrid systems~\cite{Alur01,AB01,dang04}. The formal safety analysis of hybrid systems is becoming increasingly significant with wider proliferation of automated control in circuits and systems. 
	
	An important component of any formal verification framework is the mechanism for formally specifying the design intent. In the discrete domain, formalisms based on temporal logic have been widely adopted with the use of standard assertion languages, such as SystemVerilog Assertions (SVA)~\cite{SVA} and Property Specification Language (PSL)~\cite{PSL}. Analog Mixed-Signal (AMS) extensions of assertions have been explored as well~\cite{ams-ltl1,porv01,ASL2008} and provide constructs for assertions over real-valued attributes. Assertion based verification of {HA} has also been studied~\cite{HySTL16}, while tools such as SpaceEx~\cite{spaceex11} have been used to analyze timed and hybrid models of embedded control systems using reachability analysis and model checking. However, assertions in these languages, in and of themselves limit the information carried by their Boolean outcome. Our experience is that designers want to understand what the design is doing and how it behaves, not just the success/fail scenarios. This is naturally expressed as a quantitative real-valued measure.
	
	Existing literature on quantitative specifications~\cite{porv01,a_Ouaknine2008,b_Donze2010} is assertion based, and uses metrics, with positive values indicating truth, negative values indicating violations, and robustness being described as the distance of the quantity from zero. Some metrics (such as in~\cite{c_Jha2018}) are associated with uncertainty. Assertion-based languages are designed to be flexible with respect to the assertions written but their quantitative interpretations are restricted. On the other hand, the language of {\em features}, Feature Indented Assertions (FIA)~\cite{dyFET}, is designed to be flexible on the definition of the quantity, with the set of assertions that can be expressed limited to those that are sequences of predicates and events.
		
	Many properties used in practice concerning system attributes can be intuitively expressed as features. Related work exists in learning parameters of Signal Temporal Logic (STL) properties, such as~\cite{d_JhaTSSS_2017}, wherein the authors propose learning tight bounds on parameters of STL properties from system traces. The approach can indeed be used to learn those features that can be expressed as parameters in STL properties. However, note that expressing a feature to be learned as an STL property can require the use of additional parameters, thus making the analysis more expensive.  
	
	STL property based analysis is implemented in MATLAB toolboxes such as Breach~\cite{e_Donze2010} and S-Taliro~\cite{e_Annapureddy2011}.
	These can be used for parameter synthesis, robustness monitoring and parameter sensitivity analysis. The work presented in this manuscript is largely influenced by the semiconductor industry which finds expressing properties tedious in temporal logic based languages. As a result, for digital designs, the IEEE 1800-2012 Standard SVA language is widely used across the industry for expressing assertions for the validation and verification of digital circuits. FIA is developed over the fabric of SVA, which in practice enables the more intuitive expression of real-valued quantities over system traces. 
	%{Signal Temporal Logic (STL)~\cite{porv01} and Analog Specification Language (ASL)~\cite{ASL2008}  
	
	To understand features, consider the {\em settling time} of a DC-DC Buck Regulator, defined as the time taken for the output voltage {\tt x1} to settle to below $Vr + \epsilon$ for two successive openings of the capacitor switch; $Vr$ is the rated voltage for the regulator. Booleanizing the notion of settling of {\tt x1} within 100 clock cycles in SVA, using propositional variables {\tt x1\_GE\_Vr} $\equiv$ {\tt (x1>=Vr+E)}, and {\tt swOpen} to mean the capacitor switch is open, yields the following sequence:
	
	{\scriptsize 
		\tt
		\hspace*{-0.5cm} x1\_GT\_Vr \#\#[0:100] first\_match(@(posedge swOpen)\&\& \\
		\hspace*{0cm}!x1\_GE\_Vr \#\#[0:\$] @(posedge swOpen) \&\& !x1\_GE\_Vr)
	}	
	%{\tt x1\_GT\_Vr \#\#[0:100] @(posedge state\_open) \&\& !x1\_above \#\#[0:\$] firstmatch(@(posedge state\_open) \&\& !x1\_above)}
	%In the expression, {\tt x1\_GT\_Vr}$\equiv${\tt x1>=Vr+E}, and {\tt swOpen} is true iff the capacity switch is open.
	
	The expression represents the regulator's behaviour of settling, that is the first time when the regulator's voltage output {\tt x1} is found to be less than {\tt Vr+E} for two consecutive openings of the buck regulator capacitor switch, specified as two successive capacitor switch open events, after having risen above {\tt Vr+E}. The semantics of the assertion depends on a clock and all sequence delays are in terms of this clock. A change of clock requires re-writing the assertion with delays consistent with the revised clock, thereby inviting human error. Additionally, this form of expression requires Booleanization and is non-intuitive. 
	
	The verification of a buck converter model against an expression like the one above would yield a Boolean outcome. However, the feature {\em settling time} is a real valued artifact. In FIA this is expressed by overlaying the computation of settling time over a sequence expressing the behaviour of settling. This is done using the power of {\em local variables} to store state variable values as the sequence matches, shown below in Example~\ref{ex:settlingTime}. FIA was introduced in ~\cite{dyFET}, wherein features were used to analyze systems in a simulation environment. The formal expression of features is based on the syntactic fabric of assertions, but the definition of assertions is overlayed with real valued functions that are computed over matches of underlying logical expressions. This enables the formal expression of definitions of standard features like {\em rise time}, {\em peak overshoot} and {\em settling time}, and other design specific features. %The measurement of settling time, as a feature is shown in Example~\ref{ex:settlingTime}.
	
	\begin{example}\label{ex:settlingTime}
		\emph{\textit{Settling Time:} The local variable {\tt st} is assigned in the antecedent and is used to define the feature value {\tt settlingTime} in the consequent.}
	\end{example}
	{\scriptsize
		\begin{verbatim}
		feature settlingTime(Vr,E);
		begin
		    var st;
		    (x1>=Vr+E) ##[0:$]
		        @+(state==Open) && (x1<=Vr+E), st=$time 
		            ##[0:$] @+(state==Open) && (x1<=Vr+E)
		    |-> settlingTime = st;
		end
		\end{verbatim}
	}
	%The following SVA-like sequence expression captures the intent of the settling time being within $100\mu s$:
	
	%	\begin{verbatim}
	%		{\scriptsize 
	%		\tt
	%		\hspace*{-0.4cm}(x1>=Vr+E) \#\#[0:0.0001]
	%		@+(state==Open)\&\&(x1<=Vr+E)\\
	%		\hspace*{0.5cm}\#\#[0:\$] @+(state==Open)\&\&(x1<=Vr+E)\\
	%	}		%	\end{verbatim}
	The feature in Example~\ref{ex:settlingTime} has two parameters {\tt Vr} and {\tt E} that are used later in the contained sequence expression. {\tt st} is an uninterpreted local variable of the feature that is assigned the real time at the first opening of the switch after {\tt x1>=Vr+E}, but when {\tt x1<=Vr+E}. The variable {\tt settlingTime} has the same name as the feature, and is assigned the value of the local variable when the entire sequence matches. In the sequence expression of the feature, the notable differences with SVA are the following:
	\begin{enumerate}[topsep=0.1em]\setlength\itemsep{0em}
		
		\item Predicates over real valued signals (PORVs)~\cite{porv01}, such as {\tt (x1>=Vr+E)}, are allowed. PORVs can be over real variables or over the special variable {\tt state} which refers to the name of the mode of operation of the buck regulator automaton.
		
		\item {\tt @+} is used to denote the positive crossing of PORVs. For a predicate involving variable {\tt state}, this indicates that the state is entered. Similarly {\tt @-} may be used for negative crossings of PORVs.
		
		\item All intervals of the form {\tt \#\#[a:b]} are treated as {\em real time intervals}, as opposed to intervals countable in terms of the number of clock cycles in SVA semantics. This avoids rewriting the property if the clock cycle changes.
		
	\end{enumerate}

	The repertoire of work presented in this article is rooted in the use of features for the verification and analysis of hybrid systems and consists of the following: 
	\begin{enumerate}[topsep=0em]\setlength\itemsep{0em}
		\item A methodology to compute an over-approximation of the range of feature values for all possible runs of the system.\label{item:featureCompute}		
		
		\item A methodology for finding the extremal values of the feature range through successive refinement using SMT.\label{item:cornerCase}
	\end{enumerate}
	
	Methodology~\ref{item:cornerCase} makes its first appearance in this article. Methodology 1 was first reported in~\cite{CostaD15}, where a technique for manually transforming models was outlined, but only for very specific types of features. A more general technique was later reported in~\cite{Costa16,CostaD17} with its integration into simulation flows in~\cite{Costa_VLSID_2017}. 
	%Additionally, compared to~\cite{CostaD15}, formalisms used have changed to enhance the richness of the language and the methodology for hybrid system analysis has accordingly evolved. 
	Here, at the heart of Methodology 1, we present an improved computation of the {\em product} automaton of~\cite{Costa16}. The product in~\cite{Costa16} implements conservative semantics for the feature's sequence expression, yielding a feature range that ignores some matches of the sequence-expression, and also includes matches with broader semantics than intended. The semantics of the keyword {\tt first\_match} used in the SVA sequence-expression enforces predictability in the match, by ensuring that only the earliest observation of the contained sub-sequence is matched. In this article, we extend the more general product described in~\cite{Costa16} with {\em first match region semantics} described in Section~\ref{sec:product}. In a feature, a combination of first-match and non-first-match semantics may be used (as in SVA). However, in this article, for simplicity we assume that all features are evaluated with first-match region semantics. In the past, Methodology 1 was primarily used with reachset computation tools like SpaceEx~\cite{spaceex11} to compute feature ranges. However, this is not always the best solution because in our experience the results tend to be conservative. In practice, we find that tighter feature ranges can be computed using SMT tools, at the price of longer run-times, and possibly choking if the unfolding is too large. In summary, we propose two technologies, one which is faster but coarser, and another which is slower but more precise. We present various case studies on hybrid systems from the circuit and control domains. 
	
	%An extended version of this article, with additional formalisms,  proofs, and comprehensive examples is available at {\tt arXiv:1711.00669 [cs.LO]}. 
	%-----------------------------------------------------------
	\section{Preliminaries}\label{sec:functionalModeling}
	
	Given a system defined as a {HA} $\mathcal{H}$, and a feature $F$, the objective is to find the range of valuations of $F$ over all possible runs of $\mathcal{H}$. This section presents the requisite definitions of HA, and feature semantics over runs of a HA.
	
	\subsection{Hybrid Automata}
	A hybrid automaton is defined as follows:
	\begin{definition}\label{def:hybridAutomaton}
		\textbf{Hybrid Automaton}\label{HA_def}\\
		\emph{A hybrid automaton~\cite{Alur01} is a collection $\mathcal{H}$ $=$ $(Q,$ $X,$ $Lab,$ $Init,$ $Dom,$ $Edg,$ $Act)$, where:}
		{
			\begin{itemize}[topsep=0em]\setlength\itemsep{0em}
				\item \emph{$Q$ is the set of} \textbf{discrete states}\textit{ also known as} \textbf{locations};
				%\item 
				\emph{$X$ is a finite set of real-valued variables. A {\em valuation} is a function $\nu: X \rightarrow \mathbb{R}$. Let $\mathcal{V}(X)$ denote the set of valuations over $X$};
				%\item 
				\emph{$Lab$ is a} finite set of \textbf{synchronization labels};
				%\item 
				\emph{$Init \subseteq Q \times \mathcal{V}(X)$ is a set of} \textbf{initial states};
				%\item 
				\emph{$Dom(l)$ : $Q \rightarrow 2^{\mathcal{V}(X)}$ is a \textbf{domain}. \emph{Dom(l) $\subseteq \mathbb{R}^n$} is function that assigns a set of continuous states to each discrete state $l \in Q$.}
				\item \emph{$Edg$ is a set of \textbf{edges}, also called transitions. Each edge $e=(p,a,\mu,r)$ consists of a source location $p \in Q$, a target location $r \in Q$, a synchronization label $a \in Lab$, and a transition relation $\mu \subseteq \mathcal{V}(X) \times \mathcal{V}(X)$. A transition $e$ is enabled in state $(p,\nu)$ if for some valuation $\nu' \in \mathcal{V}(X)$, $(\nu,\nu') \in \mu$. We require that for each location $p \in Q$, there be a stutter transition of the form $(p,\kappa,\mu_{ID_X},p)$, $ \mu_{ID_X} = \{(\nu,\nu)| \nu \in \mathcal{V}(X)\}$.} %We use $\mu_{ID}$ to express the fact that no state change occurs on the transition, other than those mentioned explicitly in $R(\mu)$.}
				\item \emph{$Act$  is a function that assigns to each location a (possibly infinite) set of activities, where each activity $f : \mathbb{R}^{\geq 0} \rightarrow \mathcal{V}(X)$ represents an evolution of the variables over time. The set of activities is usually defined implicitly as the set of solutions to a system of differential equations or inclusions. We denote the expression associated with the time derivative of a variable $x \in X$  in location $p \in Q$ as $flow^x_p$.}
				
				%\textbf{activity function}. The activity function, denoted as $\dot{x}$ describes how the continuous state $x \in X$ evolves over time, while in location $q$;}
				%\item \emph{{$f(q,x)$} : Q $\times X \rightarrow {X}$ is an \textbf{activity function}. The activity function, denoted as $\dot{x}$ describes how the continuous state $x \in X$ evolves over time, while in location $q$;}
			\end{itemize}
		}
		\emph{A state of $\mathcal{H}$} is given as $(p,\nu) \in Q \times \mathcal{V}(X)$. ~\hfill$\Box$
		% \begin{flushright}
		% $\Box$
		% \end{flushright}
	\end{definition}
	For valuation $\nu$, we use $\nu_{\downarrow U}$ as the projection of $\nu$ on the set of variables $U \subseteq X$. For variable $u$, $\nu_{\downarrow u}$ is the value of variable $u\in X$ in state $\nu$. Similarly for a set of valuations $\mathcal{R}$, $\mathcal{R}_{\downarrow U}$ and $\mathcal{R}_{\downarrow u}$ are respectively the projection of $\mathcal{R}$ on the variable set $U$, and variable $u$ respectively.
	%In reality the operational modes of an AMS design are known in the domain. A {\em hybrid automaton} model of the battery 	charger relies on the ability to separate the behaviour of the circuit into mutually 	exclusive modes of operation. The behaviour in each mode is expressed using first order differential equations.  
	{For an edge $e=(p,a,\mu,r)$, $e \in Edg$, $G(\mu) = \{\nu| (\nu,\nu') \in \mu\}$. $G(\mu)$ is commonly known as the {\em transition guard} and is often represented as a set of predicates over variables in $X$. Similarly $R(\mu,\nu) = \{\nu'| (\nu,\nu') \in \mu\}$, is known as the {\em reset relation}, and most often appears as a function, i.e. $R(\mu,\nu) = \nu'$, $(\nu,\nu') \in \mu$. %Although For $(U,U') \in \mu$; $U \in G(\mu)$ is called the transition {\em guard} and $R(\mu,U) = U'$ is a function denoting the transition {\em reset}.
		%When a system consists of multiple interacting components, each modeled as a HA, the set of synchronization labels for a HA are used to construct the parallel composition of HA~\cite{Alur01}, to model the system as a result of the interaction of its constituent components. 
		When a system consists of multiple interacting components, we assume that a parallel composition of automata is available prior to applying the algorithms presented in this article. The methods discussed in this article are applied on linear hybrid systems that have monotonically increasing or decreasing variable dynamics in each location. Non-linear systems can be approximated as piece-wise affine models using techniques such as {\em hybridization}~\cite{asarin2007}. Furthermore, a location with non-monotonic variable dynamics can be transformed into an equivalent model with location-wise monotonic variable dynamics. The constraint of monotonicity is used in the article to accommodate existing tools (which use {\em may}, non-urgent, semantics on transitions). With tools that use urgent semantics, this restriction can be lifted.
	}
	
	\begin{definition}\label{def:run}
		\textbf{Run of the hybrid system $\mathcal{H}$}
		
		A \textit{run} of the hybrid system $\mathcal{H}$, is a finite or infinite sequence
		\begin{mycenter}[0.1em]
			$\rho: ~~ \sigma_0 \mapsto^{t_0}_{f_0} \sigma_1 \mapsto^{t_1}_{f_1} \sigma_2 \mapsto^{t_2}_{f_2} \cdotp \cdotp \cdotp$ 
		\end{mycenter}
		of states $\sigma_i=(l_i,\nu_i)$, non-negative reals $t_i \in \mathbb{R}^{+}$, 
		and activities $f_i$ of location $l_i$, such that for all $i\geq 0$,
		\begin{enumerate}[topsep=0em]\setlength\itemsep{0em}
			\item $f_{i}(0) = \nu_i$,
			\item For all $0 \leq \theta \leq t_i$, $f_{i}(\theta) \in Dom(l_i)$,
			\item There exists an edge $(l_i,a,\mu,l_{i+1})$ such that  $(f_i(t_i),\nu_{i+1}) \in \mu$ and $\nu_{i+1} \in Dom(l_{i+1})$. 
			%{Possibly use a different notation for $q_i$ and setup a mapping from each $q_i$ to a location in $Q$}
			~\hfill$\Box$
			%\item[~~Statement 3:] With respect to $\sigma_i=(q_i,\nu_i)$, $\sigma_{i}'=(q_i,f_i(t_i))$ is the state reached after spending $t_i$ time in location $q_i$, i.e. before entering state $\sigma_{i+1}$.
			%\item[~~Statement 4:] The state $\sigma_{i+1}$ is a transition successor of the state $\sigma_{i}'=(q_i,f_i(t_i))$ and $\sigma_{i}'=(q_i,f_i(t_i))$ is a transition predecessor of $\sigma_{i+1}$. ~\hfill$\Box$
		\end{enumerate}
		%\begin{flushright}
		%$\Box$
		%\end{flushright}
	\end{definition}
	%===========================================================================
	%{Note that for state $\sigma_i$ in run $\rho$, $T(\rho,\sigma_i)=\sum_{k=0}^{i-1} t_k$ is the cumulative time at which state $\sigma_i$ is entered. The shorthand notation $T_i$ is used for $T(\rho,\sigma_i)$, when the context $\rho$ is known.}
	
	The {HA} model for a two location DC-DC Buck Regulator~\cite{buck},  in Figure~\ref{fig:ha:buck}, presents analysis complexities due to its high location switching frequency, yet is simple enough to help explain the notion of feature analysis. It has two variables, the voltage across the load ($x_1$) and the load current ($x_2$). The dynamics in each mode of operation may be found in~\cite{buck}. The notation $u'$ in a reset relation on an edge is the value of $u$ after the transition is taken. This model is used as a running example in this article along with the feature {\tt Settling Time} described in Example~\ref{ex:settlingTime}.
	
	\begin{figure}[t!]
		\centering
		\includegraphics[scale = 0.8]{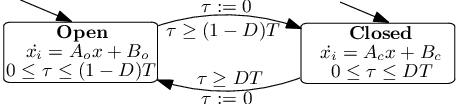}
		\caption{HA model of a DC-DC Buck Regulator~\cite{buck}} \label{fig:ha:buck}
	\end{figure}
	
	\subsection{Feature Semantics for Hybrid Automata}\label{sec:featureSemantics}
	This section presents the semantics of features which are an improved version of~\cite{Costa16}. We introduce stricter first-match semantics for feature matches in Section~\ref{sec:product}, while a more general semantic for feature-matches is presented here.
	%An example of a feature was described in the introduction.  
	The language for expressing features, FIA, uses \textit{Predicates Over Real Variables} (PORVs)\cite{porv01}. A feature is formally defined using the following syntax,\\	
	{\scriptsize
		\tt
		\begin{tabular}{ll}
			
			~~&feature $\mathcal{F}_{name}$ ($L_{p}$);\\
			~~&begin\\
			~~&\hspace*{1cm}var $\mathcal{L}$;\\
			~~&\hspace*{1cm}$S$ |-> $\mathcal{F}_{name} = \mathcal{F}$ ;\\
			~~&end\\
			
		\end{tabular}
	}\\	
	where, $\mathcal{F}_{name}$ is the feature name, ${L_p}$ and $\mathcal{L}$ are respectively the list of parameters and the list of local variables used in the body of the feature. $\mathcal{F}_{name}$ is a special variable representing the value of the feature assigned to it in the expression $\mathcal{F}$. $S$ is a sequence expression of the form,
	\begin{mycenter}[0.1em]
		$s_1 $ {\tt \#\#} $\tau_1 ~~ s_2$ {\tt \#\#} $ \tau_2  ~~ ... $ {\tt \#\#} $\tau_{n-1} ~~ s_n$
	\end{mycenter}
	and $\mathcal{F}$ is a linear function over $\mathcal{L}$ which assigns the feature value. $\tau_i$ represents a time interval, also referred to as a {\em delay operator}, and is of the form $[a:b]$, where $a,b \in \mathbb{R^+}, a\leq b$, and additionally $b$ can be the symbol $\$$, representing infinity.
	
	Sub-expressions $s_1$, $s_2$, ..., $s_n$ are each of the form "$D~\wedge~E~,~\mathcal{A}$", %:
	%\begin{mycenter}[0.1em]
	%	\mbox{$D~\wedge~E~,~\mathcal{A}$}
	%\end{mycenter}	
	where $D$ is a Boolean expression of PORVs in disjunctive normal form, $E$ is an optional event and $\mathcal{A}$ is an optional list of 
	comma-separated local variable assignments. For sub-expression $s_i$, we use $Cond(s_i) \equiv D_i \wedge E_i$ to represent a Boolean expression of PORVs and an event, and $\mathcal{A}_i = [A_i^1,A_i^2,...,A_i^k]$ is a list of $k$ local variable assignments in $s_i$. 
	The feature expression \mbox{$S$ {\tt|->} $\mathcal{F}$} is interpreted as the
	computation of $\mathcal{F}$ whenever there is a match of sequence expression $S$. We use the notation $S_i^j$, $1\leq i\leq j\leq n$ to denote the sub-sequence expression $s_i $ {\tt \#\#} $\tau_i~... $ {\tt \#\#} $\tau_{j-1} ~~ s_j$.

	Given run $\rho$ of $\mathcal{H}$, the following definitions indicate what it means for a feature to match $\rho$.
	
	%\begin{definition}
	%	\textbf{PORV of an Event E:} The PORV of an event E, denoted  $\mathcal{P}(E)$, 
	%	where $E=@^*(P)$ and $* \in \{+,-,~\}$, is the PORV $P$.~\hfill$\Box$
	%	% \begin{flushright}
	%	% $\Box$
	%	% \end{flushright}
	%\end{definition}
	
	%===========================================================================
	{
		\begin{definition}\label{def:eventMatch}
			\textbf{Event Match:}
			An event $E \equiv @^+(P)$ matches in a run $\rho: ~~ \sigma_0 \mapsto^{t_0}_{f_0} \cdotp \cdotp \cdotp \mapsto^{t_{i-1}}_{f_{i-1}} \sigma_i \mapsto^{t_{i}}_{f_{i}} \cdotp \cdotp \cdotp$ at index $i$ iff $i>0, t_{i-1} > 0, \forall_{\theta \in [0:t_{i-1})} (l_{i-1},f_{i-1}(\theta)) \nvDash P \bigwedge \sigma_i \vDash P$.
			
			We define $@^-(P) \equiv @^+(\neg P)$ and $@(P) \equiv @^+(P) \vee @^-(P)$.
			%\begin{enumerate}[topsep=0em]\setlength\itemsep{0em}
			%\item $\sigma_i \vDash \mathcal{P}(E)$ 
			%\item \label{posedge} If $E \equiv @^+(P)$, then $\exists_{T_{past}\in \mathbb{R}^{+}} ~ T_{past} < T_i$, $\forall_{T_j\in [T_{past},T_i)} \sigma_{j} \not\vDash \mathcal{P}(E)$ in run $\rho$ as defined in Def.~\ref{def:run}.
			
			%\item $\exists T_{past} < T_i, \forall t \in [T_{past}, T_i] (q_i,f_i(t-t_{i-1})) \not\vDash P(E)$
			
			%\item $\exists t < t_i, \forall t \in [0, t_i) (q_i,f_i(t)) \not\vDash P(E)$
			
			%\item $ \forall t \in [0:t_{i-1}] (q_{i-1},f_{i-1}(t)) \not \vDash P(E) \bigwedge (q_i,f_i(0) \vDash P(E)$
			%\item \label{negedge} If $E \equiv @^-(P)$, then $\exists_{T_{future}\in \mathbb{R}^{+}} ~ T_{future} > T_i$, $\forall_{T_j\in (T_i,T_{future}]} \sigma_{j} \not\vDash \mathcal{P}(E)$ in run $\rho$ as defined in Def.~\ref{def:run}.
			%\item \label{posedge} If $E \equiv @^+(P):$ $i>0, t_{i-1} > 0, \forall_{\theta \in [0:t_{i-1})} (l_{i-1},f_{i-1}(\theta)) \nvDash P \bigwedge \sigma_i \vDash P$
			
			%\item \label{negedge} If $E \equiv @^-(P):$ $i>0, t_{i-1} > 0, \forall_{\theta \in [0:t_{i-1})} (l_{i-1},f_{i-1}(\theta)) \vDash P \bigwedge \sigma_i \nvDash P$
			
			%\item If $E \equiv @(P):$ either of ~\ref{posedge} or~\ref{negedge} must hold. 
			%\end{enumerate}
			% \begin{flushright}
			% $\Box$
			% \end{flushright}
			We use the notation $\sigma_i \vDash_{\rho} E$ to denote the fact that the event E matches in the run $\rho$ at index $i$.~\hfill$\Box$
		\end{definition}
	}
	To extend predicates to be evaluated over locations of the {HA}, for state $\sigma = (l,\nu)$, a predicate $P$ can also take the form, {\tt state == $l$}, where {\tt state} is a special variable denoting the location label.
	%===========================================================================
	
	\begin{definition}\label{def:models}
		The notation $\sigma \vDash_{\rho} s$, where $s$ is treated as $Cond(s)$ and $\sigma = (l,f(t))$, is extended to conjunctions and disjunctions
		of PORVs and events recursively, as defined below. Note that $s$ does not have any delay operators.
		\begin{itemize}[topsep=0em]\setlength\itemsep{0em}
			\item $\sigma \vDash_{\rho} P$ iff $P$ is a PORV and $P$ is true for signal 
			valuation $f(t)$, or $P \equiv$ {\tt (state == $l$)}.
			\item $\sigma \vDash_{\rho} C$, where $C = P_1 \wedge P_2 \wedge ... \wedge P_n$, 
			$P_i$ is a PORV iff $\forall_{i=1}^{n} ~ \sigma \vDash P_i$.
			%\item $\sigma \vDash P_1 \vee P_2 \vee ... \vee P_n$ iff $P_i$ is a PORV and $\exists _{i=1}^{n} ~ \sigma \vDash P_i$. 
			\item $\sigma \vDash_{\rho} D$, where $D = C_1 \vee C_2 \vee ... \vee C_n$, 
			$C_i$ is a conjunction of PORVs iff $\exists _{i=1}^{n} ~ \sigma \vDash C_i$. 
			\item $\sigma \vDash_{\rho} D \wedge E$, where $D = C_1 \vee C_2 \vee ... \vee C_n$, 
			$C_i$ is a conjunction of PORVs and $E$ is an optional event, iff 
			$\sigma \vDash_{\rho} D ~ \wedge ~ \sigma \vDash_{\rho} E$. 
		\end{itemize}
		% \begin{flushright}
		% $\Box$
		% \end{flushright}
		For hybrid system $\mathcal{H}$ and sub-expression $s$, we say that $l \vDash s$ if for {some} $\sigma = (l,\nu)$, $\sigma \vDash_{\rho} s$. For sub-expression $s_j = D_j \wedge E_j$, $C^j_i$ is the $i^{th}$ conjunct term in $D_j$. $\Upsilon^j_i$ is the {\tt state} context for $C_i^j$, i.e.  $\Upsilon^j_i = l$ iff ({\tt state} $== l$) is a PORV in $C_i^j$. Similarly $\Upsilon^j$ is the state context of event $E_j$ in $s_j$.\footnote{We assume for simplicity that a conjunct in a sub-expression of the sequence does not have any contradictory `state' constraints.}
		~\hfill$\Box$
	\end{definition}
	
	%===========================================================================

	%===========================================================================
	
	\begin{definition}\label{def:seqExprMatch}
		\textbf{Match $\mathcal{M}$ of Sequence Expression $S$:}	A run $\rho: ~~ \sigma_0 \mapsto^{t_0}_{f_0} \sigma_1 \mapsto^{t_1}_{f_1} \sigma_2 \mapsto^{t_2}_{f_2} \cdotp \cdotp \cdotp$ has a match $\mathcal{M} = \langle {i_1}, \ldots, {i_n} \rangle$ of the sequence expression 
		$S=s_1 $ {\tt \#\#} $\tau_1 ~~ s_2$ {\tt \#\#} $ \tau_2  ~~ ... $ {\tt \#\#} $\tau_{n-1} ~~ s_n$, $n\geq 1$ if  $\forall_{j=1}^{n-1} i_j\leq i_{j+1}$ and the following conditions hold:
		\begin{itemize}[topsep=0em]\setlength\itemsep{0em}
			\item $\sigma_{i_1} \vDash_{\rho} s_1$, 
			\item $\sigma_{i_2} \vDash_{\rho} s_2$ and $t_{i_1} + ... + t_{i_2 - 1} \in \tau_1$,  
			\item and so on ... until,
			\item $\sigma_{i_n} \vDash_{\rho} s_n$ and $t_{i_{n-1}} + ... + t_{i_n - 1} \in \tau_n$ ,
		\end{itemize}
		%The match is denoted by $\mathcal{M} = \langle \sigma_{i_1}, \ldots, \sigma_{i_n} \rangle$. 
		Local variables associated with $s_j$ are assigned values from variable valuations in state $\sigma_{i_j}$. Note that there can be multiple matches of $S$ in $\rho$, and multiple runs of $H$ that match $S$. Each match $\mathcal{M}$ defines a feature value, denoted $Eval(\mathcal{M},\rho,\mathcal{F})$, computed as the value of feature expression $\mathcal{F}$ over values of the local variables assigned during match $\mathcal{M}$ in run $\rho$.  ~\hfill$\Box$
		% \begin{flushright}
		% $\Box$
		% \end{flushright}
	\end{definition}
	\begin{definition}\label{def:featureComputation}
		\textbf{Feature Range of a Hybrid Automaton:} Given a feature sequence expression $S$ and a feature
		computation function $\mathcal{F}$, that computes the value of $F_{name}$, the feature range $[\mathcal{F}_{min}, \mathcal{F}_{max}]$ of a hybrid automaton $H$ is computed as follows:
		\begin{itemize}[topsep=0em]\setlength\itemsep{0em}
			\item $\mathcal{F}_{min}$ = $\min \{~Eval(\mathcal{M},\rho,\mathcal{F})~|~ \rho$ is a run in $H$ and $\mathcal{M}$ is a match of $S$ in $\rho$ $\}$
			\item $\mathcal{F}_{max}$ = $\max \{~Eval(\mathcal{M},\rho,\mathcal{F})~|~ \rho$ is a run in $H$ and $\mathcal{M}$ is a match of $S$ in $\rho$ $\}$
		\end{itemize}
		% \begin{flushright}
		% $\Box$
		% \end{flushright}
		max and min are computed over all matches in all runs $\rho$ of $H$.~\hfill$\Box$
	\end{definition}	
	
	\section{Method-1: Feature Transformations}\label{sec:featureTransformations}
	The problem of feature analysis can be formally defined as follows:
	Given a {HA} $\mathcal{H} = (Q, X, Lab, Init, Dom, Edg, Act)$ and a feature ${F}$, we wish to compute $\mathcal{F}_{min}$
	and $\mathcal{F}_{max}$ over all runs $\rho$ of $\mathcal{H}$.
	
	The methodology consists of the three steps shown in Figure~\ref{fig:methodology1}. Compared to~\cite{Costa16}, the Feature Automaton definition is re-written to better capture the feature intent of Section~\ref{sec:featureSemantics}, and in Section~\ref{sec:product} we present an improved construction of the product automaton. 
	
	\subsection{Feature Automaton Construction}
	The feature automaton is a monitor automaton which is similar to a HA, but allows guards to be written with predicates over location labels and events. Additionally, unlike the HA in Definition~\ref{def:hybridAutomaton}, a feature automaton also has an accept location.
	
	Given a {HA} $\mathcal{H}$ $=$ $(Q_H,$ $X_H,$ $Lab_H,$ $Init_H,$ $Dom_H,$ $Edg_H,$ $Act_H)$, a feature $F$ with local variable set $\mathcal{L} = \{l_0,l_1,...,l_m\}$, 
	sequence expression $S=s_1 $ {\tt \#\#} $\tau_1 ~~ s_2$ {\tt \#\#} $ \tau_2  ~~ ... $ {\tt \#\#} $\tau_{n-1} ~~ s_n$ and feature computation expression $\mathcal{F}$ that assigns the feature value to $\mathcal{F}_{name}$, we construct the feature automaton as follows:
	\begin{definition}\label{def:featureAutomaton}
		\textbf{Feature Automaton}: \label{FA_def}	\emph{A Feature Automaton for HA $\mathcal{H}$ is a collection $M_F = (Q, Z, X, V, C, E, Init, q_{\mathcal{F}})$, where:} 
		
		\begin{itemize}[topsep=0em]\setlength\itemsep{0em}
			\item \emph{$ Q = \{q_{1}, q_{2}, ... , q_{{n+1}}, q_{\mathcal{F}}\} \cup Z$ is the set of 
				feature locations. Intuitively, location $q_{i}$ is reached when
				the sequence expression has matched upto $s_{i-1}$ and it awaits the 
				match of $s_{i}$ within the time interval $\tau_{i-1}$. $q_{\mathcal{F}}$ is the {\em accept} location;} 
			\item \emph{Z = $\{q_{i,j} |~ |\mathcal{A}_i|>1, 1\leq i \leq n, 1\leq j \leq |\mathcal{A}_i|-1\} \cup \{q_{n+1}, q_{\mathcal{F}}\}$  is a set of
				\textbf{pause locations} where time does not progress. If a subexpression has multiple assignment statements, then a sequence of pause locations are added corresponding to all but the first assignment, to capture the order in which these assignments are made. Pause locations 
				$q_{i,j}$ are added if $|\mathcal{A}_i| > 1$.  State $q_{i,j}$ is reached when the $j^{th}$ assignment of $s_i$ has been executed;
				%$\forall_{i=1}^n ~ (|\mathcal{A}_i|>1 \leftrightarrow 	\forall_{j=1}^{|\mathcal{A}_i| - 1} ~ (q_{iR_j} \in Z))$, $\{q_{n+1}, q_{\mathcal{F}}\} \subseteq Z$.
			}
			\item $X = X_H$; 
			%\item 
			$V$ is the set of feature variables, $V = \mathcal{L} \cup \{\mathcal{F}_{name}\}$;
			%\item 
			\emph{C = $\{t, lt\}$ is a set of \textbf{timers} where $t$  measures cumulative time along a match, and $lt$ is a location timer, measuring time spent in every location. All timers are initially 0;}
			%\item $Init$ = $q_{s_1} \times \eta_{Init}(X)$. All variables of the automaton must be initialized before static analysis can be carried out. $\eta_{Init}(X)$ is set to a known constant vector dependent on the use of the variables in the feature computations step.
			%\item $Dom(q_{s_i},x) = [0,\sqcap(\tau_{i-1})]$ where $\sqcap(A)$ represents the supremum of set $A$; for clock variable 
			%$x$. All non-clock variables have unrestricted domains.
			\item \emph{We augment the assignment list $\mathcal{A}_i$ for sub-expression $s_i$ with the assignment {\tt lt:=0} for the {\em location timer} $lt$.}
			\item \emph{An event of the form $@^*(x \sim a)$ is associated with PORVs as follows:
				\begin{itemize}
					\item $@^+(x \geq a)$ is associated with $(x==a)$ and $flow_q^x > 0$ 
					\item $@^-(x \leq a)$ is associated with $(x==a)$ and $flow_q^x < 0$ 
			\end{itemize}}
			\item \emph{$E\subseteq Q\times \mathcal{V}(Q_H \cup X\cup V\cup C) \times Q$ is the set of edges defined by the following rules: 
				\begin{itemize}
					\item $\forall_{1\leq i \leq n}(|\mathcal{A}_i| = 1) \rightarrow (q_{i},\mu_i,q_{{i+1}}) \in E$
					%\item $ \forall_{1\leq i \leq n}(|\mathcal{A}_i|> 1) \rightarrow$ 
					%\\\hspace*{1cm} $(q_{{i}},\mu_{i},q_{i,1}) \in E ~ \wedge $
					%\\\hspace*{1cm} $\forall_{1\leq j < |\mathcal{A}_i|-1}(q_{i,j},\mu_{i,j},q_{i,j+1}) \in E ~ \wedge $
					%\\\hspace*{1cm} $(q_{i,{|A_{i-1}|}},\mu_{i,{|\mathcal{A}_{i-1}|}},q_{{i+1}}) \in E $.
					\item $ \forall_{1\leq i \leq n}(|\mathcal{A}_i|> 1) \rightarrow$ 
					$(q_{{i}},\mu_{i},q_{i,1}) \in E ~ \wedge $
					$\forall_{1\leq j < |\mathcal{A}_i|-1}$ $(q_{i,j},\mu_{i,j},q_{i,j+1})$ $\in E$ $\wedge$ 
					$(q_{i,{|A_{i-1}|}},\mu_{i,{|\mathcal{A}_{i-1}|}},q_{{i+1}}) \in E $.
					\item $\forall_{1 \leq i < n} ~\mu_i = (Cond(s_i) \wedge (lt\in \tau_{i}) \wedge A_i^1 \wedge \{{\tt lt'==0}\})$; $\mu_n = (Cond(s_n) \wedge A_n^1)$
					\item $\forall_{1 \leq i < n} ~\forall_{1 \leq j \leq |\mathcal{A}_i|-1} ~ \mu_{i,j} = true \wedge A_i^{j+1}$ 		
					%\item $\forall_{1 \leq i < n} ~ G(\mu_i) \equiv (Cond(s_i) \wedge (lt\in \tau_{i}))$, $G(\mu_n) \equiv Cond(s_n)$,\\
					%$R(\mu_i,\nu) \equiv A_i^1 \wedge \{{\tt lt:=0}\}$			
					%\item $\forall_{1 \leq j \leq |\mathcal{A}_i|-1} ~ G(\mu_{i_j}) \equiv true$, $R(\mu_{i_j},\nu) \equiv A_i^{j+1}$ 		
					\item $(q_{{n+1}},\mu_{F},q_{\mathcal{F}}) \in E$, where $\mu_F = \{\mathcal{F}_{name}' == \mathcal{F}\}$
					%\item $R(\mu_i,\nu) \equiv A_i^1$ 
					\item For relation $\mu$, $\omega$ denotes the projection of the relation $\mu$ onto $\mathcal{V}(X\cup V \cup C) \times \mathcal{V}(X\cup V \cup C)$.
			\end{itemize}}
			\item {\em $Init = \{q_1\} \times [0]^{|V\cup C}$ is the set of initial states.}
		\end{itemize}
		\emph{A state of $M_F$} is given as $(q,\nu) \in Q \times \mathcal{V}(X\cup V\cup C)$.~\hfill$\Box$
	\end{definition}
	
	\begin{figure}[t]
		\centering
		\includegraphics[scale=0.75]{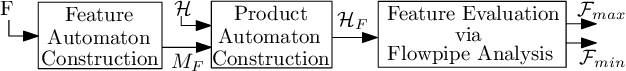}
		\caption{Methodology-1: Feature Transformations}\label{fig:methodology1}
	\end{figure}

	\begin{figure}[t]
		\centering
		\includegraphics[scale=0.9]{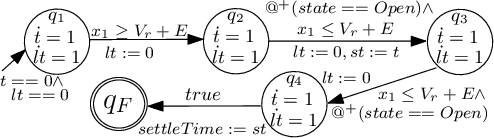}
		\caption{Feature Automaton for feature {\tt Settling Time}.}\label{fig:featureAutSettleTime}
	\end{figure}

	For feature {\tt Settling Time} of Example \ref{ex:settlingTime}, the FA is shown in Figure~\ref{fig:featureAutSettleTime}. The FA has two timer variables, $t$ for measuring time along the entire run, and $lt$ for measuring the time spent in each location of the FA, indicative of delays separating subexpressions in the feature sequence expression. Each location of the FA represents the match of some feature sub-expression. $q_i$ represents that the temporal sequence of events and PORVs leading up to (but not including) the $i^{th}$ sub-expression has been observed. The transition between $q_i$ and $q_{i+1}$ is guarded by the Boolean expression of PORVs and events corresponding to the $i^{th}$ sub-expression; the associated set of assignments to local variables are computed along this transition. Progressing	along  transitions between locations of the automaton corresponds to matching each sub-expression. %, keeping temporal constraints in mind, and assigning values computed to local variables as we move forward. For a sequence expression having $n$ sub-expressions, a match of the $(n-1)^{th}$ sub-expression places the automaton in location $q_{n}$, where it waits for the last sub-expression to match.
	When the $n^{th}$ sub-expression matches, the entire sequence expression has matched and the automaton  transitions to state $q_{n+1}$. At $q_{n+1}$, all local variables hold values assigned to them along the match, and the feature is computed along the unguarded transition from $q_{n+1}$ to $q_{\mathcal{F}}$, the accept location of the automaton.

	%Given a FA $M_F$ for feature $F$, by construction, a run $\rho$ of the hybrid automaton $\mathcal{H}$ is accepted by $M_F$ iff there is a match $\mathcal{M}$ of $F$.
	
	%In the sections that follow, we elaborate these three steps.
	
	\begin{definition}\textbf{Acceptance of run $\rho$ of $\mathcal{H}$ by $M_F$:}
		Run $\rho: ~~ \sigma_0 \mapsto^{t_0}_{f_0} \sigma_1 \mapsto^{t_1}_{f_1} \sigma_2 \mapsto^{t_2}_{f_2} \cdotp \cdotp \cdotp $ of $\mathcal{H}$ is accepted by feature automaton $M_F$ for feature $F$ with sequence expression $s_1 $ {\tt \#\#} $\tau_1 ~~ s_2$ {\tt \#\#} $ \tau_2  ~~ ... $ {\tt \#\#} $\tau_{n-1} ~~ s_n$ iff $\exists$ $\sigma_{i_1},$ $\sigma_{i_2},...,\sigma_{i_n},$ such that $i_j \leq i_{j+1}$, where, $\sigma_{i_1} \vDash s_1$;  
		\begin{itemize}[topsep=0em]\setlength\itemsep{0em}
			%\item $\sigma_{i_1} \vDash s_1$ 
			\item $\sigma_{i_2} \vDash s_2 \bigwedge (lt = \sum_{k=i_1}^{i_2 - 1} t_k) \in \tau_1$; $\ldots$
			%\item $\sigma_{i_3} \vDash s_3 \bigwedge (lt = \sum_{k=i_2}^{i_3-1} t_k) \in \tau_2$\\
			%$\cdots$
			\item $\sigma_{i_n} \vDash s_n \bigwedge (lt = \sum_{k=i_{n-1}}^{i_{n}-1} t_k) \in \tau_{n-1}$    ~\hfill$\Box$
		\end{itemize}
	\end{definition}
	
	\begin{theorem}
		Given a feature, $F$ in FIA, for {HA} $\mathcal{H}$,
		the feature automaton $M_F = (Q, Z, X, V, C, E, q_{\mathcal{F}})$  for $F$ correctly captures the following feature 	semantics:\label{theorem:FeatureAcceptance}
		\begin{enumerate}[topsep=0em]\setlength\itemsep{0em}
			\item[A] If a run $\rho$ of $\mathcal{H}$ yields a match $\mathcal{M}$, then the run $\rho$ is
			accepted by feature automaton $M_F$ with the same valuation as $Eval(\mathcal{M},\rho,\mathcal{F})$.
			\item[B] If a run $\rho$ of $\mathcal{H}$ is accepted by $M_F$ with valuation $\gamma$, then $\rho$ has a match $\mathcal{M}$, such that $Eval(\mathcal{M},\rho,\mathcal{F}) = \gamma$.
		\end{enumerate}
	\end{theorem}	
	{\setlength{\topsep}{1pt}
		\begin{proof}
			We prove the theorem in two parts as follows:
			
			{\em Part A:} Let $\rho: ~~ \sigma_0 \mapsto^{t_0}_{f_0} \cdotp \cdotp \cdotp \mapsto^{t_{i_1-1}}_{f_{i_1-1}} \sigma_{i_1} \mapsto^{t_{i_1}}_{f_{i_1}} \cdotp \cdotp \cdotp \mapsto^{t_{i_2-1}}_{f_{i_2-1}} \sigma_{i_2} \mapsto^{t_{i_2}}_{f_{i_2}} \cdotp \cdotp \cdotp $ be a run of $\mathcal{H}$ that matches the sequence expression $S=s_1 $ \#\#$\tau_1 ~~ s_2$ \#\# $\tau_2  ~~ ... $ \#\#$\tau_{n-1} ~~ s_n$, with match $\mathcal{M} = \langle {i_1}, \ldots, {i_n} \rangle$. Let $M_F$ be the feature automaton constructed for feature $F$.
			
			The initial location of $M_F$ is $q_{1}$. In the prefix $\sigma_0 \mapsto^{t_0}_{f_0} \cdotp \cdotp \cdotp \mapsto^{t_{i_1-1}}_{f_{i_1-1}} \sigma_{i_1}$ of $\rho$, $\sigma_{i_1} \vDash_\rho s_1$ and G($\mu_1$) = $s_1$, hence the state $q_{2}$ is reachable in the feature automaton with state $\sigma_{i_1}$ of $\rho$, with associated assignments to local variable made along the transition. For configuration $\langle q_{{j}},\sigma_{i_j}\rangle$ reachable in $M_F$, $\langle q_{{j+1}},\sigma_{i_{j+1}}\rangle$ is reachable for all $ 1\leq j\leq n$. Given that $\langle q_{{j}},\sigma_{i_j}\rangle$ is reachable, $(q_{{j}},\mu_j,q_{{j+1}})$ is an edge in $M_F$, and $\mathcal{M}$ is a match, we have $\sigma_{i_j} \vDash_{\rho} s_{j}$ and $G(\mu_{j}) \equiv s_{j}$ $\wedge$ $(t_{i_{j-1}} + ... + t_{i_j - 1}) \in \tau_{j}$, $\langle q_{{j+1}},\sigma_{i_j}\rangle$ is reachable, via one or more transitions through pause states. Now, since $\sigma_{i_j} \mapsto^{t_{i_j}}_{f_{i_j}} \cdotp \cdotp \cdotp \mapsto^{t_{i_{j+1}-1}}_{f_{i_{j+1}-1}} \sigma_{i_{j+1}}$, configuration $\langle q_{j+1}, \sigma_{i_{j+1}}\rangle$ is also reachable. Inductively, when configuration $\langle q_{{n}},\sigma_{i_{n}}\rangle$ is reached, $\langle q_{{n+1}},\sigma_{i_{n}}\rangle$ is also reachable. 	Since $G(\mu_{n+1}) = true$, $\langle q_{\mathcal{F}},\sigma_{i_{n}}\rangle$ is reachable. Along each transition, resets corresponding to local variable assignments 
			appropriately update the values of the local variables which form part of the feature automaton state. The feature
			expression is computed on the transition to location $q_{\mathcal{F}}$.
			
			{\em Part B:} The run $\rho$ of $\mathcal{H}$ is accepted by $M_F$. Therefore, $\exists ~\sigma_{i_1},$ $\sigma_{i_2},$ $...,$ $\sigma_{i_n}$ in $\rho$ such that $i_j \leq i_{j+1}$ and  $\sigma_{i_1} \vDash_{\rho} s_1$, $\sigma_{i_2} \vDash_{\rho} s_2 ~ \wedge ~ t_{i_1} + ... + t_{i_2 - 1} \in \tau_1$, and so on ... until, $\sigma_{i_n} \vDash_{\rho} s_n ~ \wedge ~  t_{i_{n-1}} + ... + t_{i_n - 1} \in \tau_n$. The feature valuation computed over state $\sigma_{i_n}$, on acceptance, is $\gamma$. By Definition~\ref{def:seqExprMatch}, $\mathcal{M} = \langle {i_1},{i_2},...,{i_n} \rangle$ is a match of $\mathcal{F}$ with valuation $Eval(\mathcal{M},\rho,\mathcal{F}) = \gamma$ in the feature range of Definition~\ref{def:featureComputation}.		
	\end{proof}}

	\subsection{Product Automaton Construction}\label{sec:product}
	The product construction of the {HA} $\mathcal{H}$ and the FA $M_F$ yields a special type of automaton. In the classical product construction, non-determinism present in the component automata carries over to the product. The more traditional product construction, defined in~\cite{Costa16}, is conservative and reports unintentional matches, while at times missing matches that were otherwise intended. It is thus not complete. In this article a non-standard product is defined that has clearer semantics for matches than that described in earlier work. To accomplish this, we introduce the notion of first-match region semantics. 
	
	We explain this with an example. Consider the designer's intention to specify the pattern, {\em "P1 is true and thereafter P2 is true"}, and measure the time delay between the two. This translates to the sequence-expression ``{\tt P1,t1:=\$time \#\#[0:\$] P2,t2:=\$time}'', where {\tt P1} and {\tt P2} are PORVs over analog signals {\tt x} and {\tt y} respectively, and the feature computation is {\tt (t2-t1)}. The truth intervals of the PORVs are shown as $r_1$, $r_2$, $r_3$, $r_4$ and $r_5$ in Figure~\ref{fig:fm_wave}. With the semantics of~\cite{Costa16}, the maximum feature value would be $\Delta_1$ (matching points in $r_1$ with points in $r_5$). But the intent is to match points where {\tt P2} is true immediately subsequent to points where {\tt P1} is true, giving a maximum feature value of $\Delta_2$. This cannot be captured by the semantics of ~\cite{Costa16}. First match region semantics matches $r_1$ only with $r_2$, giving a maximum feature value of $\Delta_2$. Region $r_4$ matches with $r_5$. Region $r_3$ doesn't contribute to any match.

	\begin{figure}[t]
		\centering
		\includegraphics[scale=0.9]{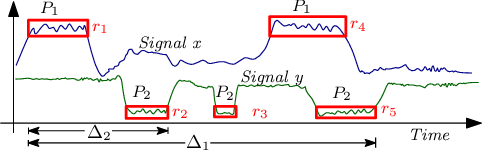}
		\caption{An illustration of first-match region semantics.}\label{fig:fm_wave}
	\end{figure}

	\begin{definition}{\bf First-match Region Semantics}\label{def: firstMatch}
		
		Given a sequence expression, $S$$=$$s_1${\tt \#\#}$\tau_1 ~ s_2${\tt \#\#}$\tau_2 ...${\tt \#\#} $\tau_{n-1}~s_n$, and $\mathbb{M}=\{$$\langle {1_1},$${1_2},$$\ldots,$${1_n} \rangle,$$\langle {2_1}, {2_2},$$\ldots, {2_n} \rangle$ $ \ldots \}$, the set of all matches of $S$ in \textit{run} $\rho: ~~ \sigma_0 \mapsto^{t_0}_{f_0} \sigma_1 \mapsto^{t_1}_{f_1} \sigma_2 \mapsto^{t_2}_{f_2} \cdotp \cdotp \cdotp $ of the hybrid system $\mathcal{H}$, $ \langle {i_1}, {i_2}, \ldots, {i_n} \rangle \in {\mathbb M}$ follows first match region semantics iff:\\  
		\hspace*{0.5cm}$\forall_{j=2}^{n} \exists_{\theta_l \in [0,T_{i_j}]}~ \exists_{\theta_r \in [T_{i_j},\infty)}$ \\
		\hspace*{1cm}$(\forall_{t<\theta_l} \langle {i_1},\ldots, {k} \rangle$ is not a match for $S_{1}^{j}, T_k = t)$ and \\
		\hspace*{1cm}$(\forall_{\theta_l\leq t'\leq \theta_r} \langle {i_1},\ldots, {k'} \rangle$ is a match for $S_{1}^{j}, T_{k'} = t')$.
		
		where, $T_m = \Sigma_{z=0}^{m-1} t_z$.
		%$\forall_{j=2}^{~n} ~ ( \mathbb{T}( \sigma_{i_j} ) - \mathbb{T}( \sigma_{i_{j-1}} ) ) \leq ( \mathbb{T}( \sigma_k )  - \mathbb{T}( \sigma_{i_{j-1}} ) ), k\geq i_j, \sigma_k \vDash s_j$.
		~\hfill$\Box$
	\end{definition}
	%Future use of the term match in this article indicate a match that satisfies Definition~\ref{def: firstMatch}.
	%The restrictions on matches imposed by Definition~\ref{def: firstMatch} require the addition of constraints to limit the non-determinism in the resulting product automaton. 
	The product definition presented here ensures that only runs following first match semantics reach the final location in $M_F$. The exclusion of other runs that would have matched in a traditional product is intentional, and imposed to accurately embody first match semantics for feature computation in FIA in the generated product. Additionally, the FA doesn't follow the traditional structure of an observer automaton for verification. These reasons taken together motivate the need for a non-standard product construction.
	
	We denote the valuation of variables in the state $\sigma$ in a run $\rho$ as $\eta(\sigma)$, and the valuation of the variable $v$ in the state $\sigma$ as $\eta(\sigma[v])$. %We additionally extend the notation $\sigma \vDash \varphi$, where $\varphi$ is a Boolean expression of PORVs and events, to $q \vDash \varphi$, where $q$ is a location, according to Definition~\ref{def:models}.
	{
		\begin{definition}\label{def:LSHA}
			\textbf{Level Sequenced Hybrid Automaton (LSHA) - $M_F \bowtie \mathcal{H}$ :}\label{def:featureHybridSynch}
			The product of feature automaton $M_F = (Q_S,$$Z,$$X_H,$$V,$$C,$$E,$ $Init_S, q_{\mathcal{F}})$
			and {HA} $\mathcal{H}$ = $(Q_H,$ $X_H,$ $Lab_H,$ $Init_H,$  $Dom_H$, $Edg_H$, $Act_H)$, is defined as 
			the {HA} $\mathcal{H}_F$ = $(Q_F,$ $X_F,$ $Lab_F,$ $Init_F,$ $Dom_F,$ $Edg_F,$ $Act_F)$ where,
			\begin{itemize}[topsep=0em]\setlength\itemsep{0em}
				\item \emph{$Q_F = Q_H \cup \{q_{\mathcal{F}}\} \cup \{Z \times Q_H\} $};
				
				\item \emph{$X_F = X_H \cup {V \cup C \cup \{level\}}$};
				
				\item \emph{$Lab_F = Lab_H$, is the finite set of \textbf{synchronization labels};}
				
				\item \emph{$Init_F = Init_H \times Init_S \times \{level==0\}$};
				
				\item \begin{tabular}[t]{@{\hspace{-1pt}}l c l }
					$Dom_F(l)$ 
					& = & $Dom_H(l) \times \mathbb{R}^{|V \cup C|+1}$ if $l \in Q_H$,\\
					& = & $\mathbb{R}^{|X_F|}$ if $ l \in { \{Z \times Q_H\}} \cup \{q_{\mathcal{F}}\}$
				\end{tabular}
				
				\item \emph{$Edg_F \subseteq Q_F \times Lab_F \times \mu_{\mathcal{V}(X_F) \times \mathcal{V}(X_F)} \times Q_F $ is defined by the following {rules}, \emph{where $l \in Q_H$ and $q_i,q_{i'} \in Q_S/Z$, $\mu_H$ and $\mu_i$ are the transition relations $\mu_H \subseteq \mathcal{V}(X_H) \times \mathcal{V}(X_H)$ and $\mu_i \subseteq \mathcal{V}(Q_H\cup X_H \cup V \cup C) \times \mathcal{V}(Q_H\cup X_H \cup V \cup C)$, with $\omega_i$ as defined in Definition~\ref{def:featureAutomaton}. 
				The relation $l\vDash \mu_i$, to be read $\mu_i$ is applicable in $l$, is true iff $\exists_{C^i_j \in s_i}  \Upsilon^i_j == l$; and for edge $e=(l,a,\mu_H,l')$, $e \vDash \mu_i$ iff for $s_i = D_i ~\wedge ~E_i$ either $E_i \equiv @^-(state == l)$ or $E_i \equiv @^+(state == l')$.
					}:}
				
				{\small
					%\begin{center}
					%\begin{tabular}{c}%p{3.2cm} p{3.2cm}}
					%For when $level= 0$:
					\[
					{	
						{l} \xhookrightarrow[{\mu_H}]{a} l'
					} 
					\over 
					{
						l 
						\xhookrightarrow[\mu_H \times \mu_{ID_{\{V\cup C\}}} \times \{ (level == 0) \wedge (level':=level)\}]{a}
						l'
					} \tag{13.1} \label{eq:13.1}
					\]
					%For when $level \neq 0$:
					\[
					{	
						{l \xhookrightarrow[\mu_H]{a} l'} \bigwedge~
						{q_i \xhookrightarrow[\mu_i]{} q_{i'}} \bigwedge ~
						l  \nvDash \mu_i
					} 
					\over 
					{
						l 
						\xhookrightarrow[\mu_H \times (\mu_{ID_{\{level\}\cup V\cup C}}\cap \{level==i-1\})]{a} 
						l'
					}
					\tag{13.2} \label{eq:13.2-1}
					\]
					\[ 
					{
						q_i \xhookrightarrow[\mu_i]{} q_{i'}  ~ \bigwedge ~ 
						l \vDash \mu_i 
					} 
					\over 
					{
						l \xhookrightarrow[\mu_{ID_{X_H}} \times \omega_i \times  \{(level==i-1) \wedge (level':=i)\} ]{} l
					}
					\tag{13.3} \label{eq:13.3}
					\]
					%For pause locations, location $l\in Q_H$,  $e=(l,a,\mu_H,l') \in Edg_H$:%location $l \vDash G(\mu_i)$ for $q_i \xhookrightarrow{\mu_i} q'_f$, $q_i \notin Z$, $q'_f \in Z$:
					\[
					{
						q_i \xhookrightarrow[\mu_{i}]{} q_{i,1}  ~ \bigwedge ~ 
						(l \vDash \mu_i \vee  e\vDash \mu_i)~  \bigwedge ~ 
						q_{i,1} \in Z 
					} 
					\over 
					{
						l \xhookrightarrow[\mu_{ID_{X_H}} \times \omega_i \times  \{(level==i-1) \wedge (level':=i)\}]{} (q_{i,1},l)%{^{\langle q_H \rangle}}
					}
					\tag{13.4} \label{eq:13.4}
					\]
					\[ 
					{
						q_{i,j} \xhookrightarrow[\mu_{i,j}]{} q_{i,j+1}  ~ \bigwedge ~
						q_{i,j}, q_{i,j+1} \in Z 
					} 
					\over 
					{
						(q_{i,j},l)
						\xhookrightarrow[\mu_{ID_{X_H \cup \{level\}}} \times \omega_{i,j}]{} (q_{i,j+1},l)%{^{\langle q_H \rangle}}
					}
					\tag{13.5} \label{eq:13.5}
					\]
					\[ 
					{
						q_{i,j} \xhookrightarrow{\mu_{i,j}} q_{i'}  ~ \bigwedge ~
						q_{i,j} \in Z ~  \bigwedge ~ q_{i'} \notin Z \bigwedge e\nvDash \mu_{i}
					} 
					\over 
					{
						(q_{i,j},l)
						\xhookrightarrow[\mu_{ID_{X_H \cup \{level\}}} \times \omega_{i,j}]{} l
					}
					\tag{13.6} \label{eq:13.6}
					\]
					\[ 
					{
						q_{i,j} \xhookrightarrow{\mu_{i,j}} q_{i'}  ~ \bigwedge ~
						q_{i,j} \in Z ~  \bigwedge ~ q_{i'} \notin Z \bigwedge e\vDash \mu_i
					} 
					\over 
					{
						(q_{i,j},l)
						\xhookrightarrow[\mu_{ID_{X_H \cup \{level\}}} \times \omega_{i,j}]{} l'
					}
					\tag{13.7} \label{eq:13.7}
					\]
					
					%\end{tabular}
					%\end{center}
				}
				%\emph{where $l \in Q_H$ and $q_i,q_{i'} \in Q_S$, $\mu_H$ and $\mu_i$ are the transition relations $\mu_H \subseteq \mathcal{V}(X_H) \times \mathcal{V}(X_H)$ and $\mu_i \subseteq \mathcal{V}(X_H) \times V \times C \times \mathcal{V}(X_H) \times V \times C$. \\
				%$l\vDash \mu_i$ is false iff $s_i\vDash${\tt (state==p)}, {\tt l$\neq$p}; and for edge $e=(l,a,\mu_H,l')$, $e \vDash \mu_i$ iff for $s_i = D_i ~\wedge ~E_i$ either $E_i \equiv @^-(state == l)$ or $E_i \equiv @^+(state == l')$.
				%} %\emph{ $\mu_{ID}$ represents the identity relation between variable values of variables in  $X_H \cup V \cup C$.}
				
				\item \emph{The function $Act_F$ assigns a set of activities to each location. The expression associated with the flow ${flow_F}^q_x : \mathbb{R}^{\geq0}\rightarrow \mathbb{R}^n$ for each $x \in X_F$ in location $l \in Q_F$ is defined as follows:
					\begin{mycenter}[0.1em]
						\begin{tabular}{l c l }
							${flow_F}^l_x$ 
							& = & 0 for each $x \in X_F$ if $l \in Z \times Q_H$,\\
							& = & 0 for each $l \in Q_F$ if $x \in V$,\\						
							& = & 1 for each $l \in Q_F$ if $x \in C$,\\
							& = & ${flow_H}^l_x$ if $l \in Q_H \bigwedge x\in X_H$
						\end{tabular}
					\end{mycenter}
				}
			\end{itemize}        
			{\em In order to enforce first match semantics, for each $q_i \xhookrightarrow[\mu_i]{} q_{i'} \bigwedge q_H  \vDash \mu_i$, $q_i, q_{i'} \in Q_S$, $i>1$, $q_H,q_H' \in Q_H$, at level $i-1$, $q_H$ is replaced in $Q_F$ according to the transformation in Figure~\ref{fig:splitting}. Herein, $q_{H_i}$ is identical to $q_H$ and differs only in the invariant as shown. $\overline{G}(\mu_i)$ is the closure of the compliment of conditions satisfying $s_i$.~\hfill$\Box$ %defined as follows:
			%	\begin{mycenter}[0.1em]
			%		\flushleft{$\widetilde{G}(\mu) = \{ \nu| \nu \notin G(\mu)  \bigwedge%$} \\ 
					%\flushright{$
			%		(\nu' \in G(\mu) \wedge ((\nu\downarrow_x < \nu'\downarrow_x \wedge {flow_H}^l_x >0)$} \\ 
			%	\flushright{$\vee(\nu\downarrow_x > \nu'\downarrow_x \wedge {flow_H}^l_x<0) ))\}$}
			%	\end{mycenter}
				
			}%$\widetilde{G}(\mu)$ is the {\em non-strict} complement of the guard for relation $\mu$, i.e. compliments observe De-Morgan's law, except that the compliment of $x \geq a$ is $x \leq a$ and vice-versa.}
		
		\end{definition}}
	
	\begin{figure}[t]
	\centering
	\includegraphics[scale=0.8]{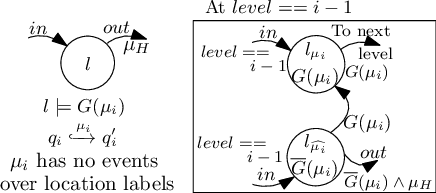}
	\caption{Splitting for location $l \in Q_H$ when $l \vDash G(\mu_i)$, $in$ and $out$ represent incoming and outgoing transitions of $q_H$.}
	\label{fig:splitting}
	\end{figure}
	
	The behaviour asserted by the feature sequence-expression is built into the product automaton called a {\em level-sequenced hybrid automaton} (LSHA). In the LSHA, the level is a syntactic structure derived from the sequence-expression. The value of the variable {\em level} indicates how much of the sequence expression has matched. The variable $level$ is set to $i$ when sub-expression $i$ matches. Transitions from one level to another assign an appropriate value to $level$ indicative of the subscript of the sub-expression matched, while also executing assignments to local variables associated with the match. Initially, when $level$ is 0, corresponding to location $q_{1}$ of $M_F$, the automaton waits for a match of $s_1$. When $s_1$ matches, $level$ is non-deterministically incremented. Due to first-match semantics, a constrained form of non-determinism is applicable when $level>0$. The non-determinism in the control allows computation of a continuum of matches. When the feature has matched (ending with the match of the last sub-sequence),  the feature is computed and control moves to the final location of $M_F$.
	%The LSHA is a multi-level automaton, with the current level indicated by the value of the variable $level$, introduced into the automaton to keep track of how much of the feature's sequence expression has matched. Initially, when at level 0, corresponding to the location $q_{1}$ of $M_F$, the automaton waits for a match of the first sub-sequence $s_1$. When $s_1$ matches, the control non-deterministically moves up one level or stays in the same level. A constrained form of non-determinism, due to first match region semantics is applicable at higher levels. When moving up a level, any assignments to the local variables associated with the match are made. The non-determinism in the control allows us to identify a continuum of matches. When the entire feature has matched (ending with the match of the last sub-sequence), control moves to the final location of the FA and the feature is computed. We collapse levels in which a location $q_H$ of $\mathcal{H}$ may be in, using the variable $level$. The variable $level$ remains constant while in a location, but may be changed discretely on a transition. The variable $level$ is reset to $i$ when sub-expression $i$ matches. Transitions from one level to another assign an appropriate value to $level$ indicative of the subscript of the sub-expression matched. Since the final location is a terminating location, only one such location is required. In the final location, no variables evolve. Thus, one copy of $\mathcal{H}$'s locations suffices to represent all levels.
	
	\begin{figure}[t]
	\centering
	\includegraphics[scale=0.8]{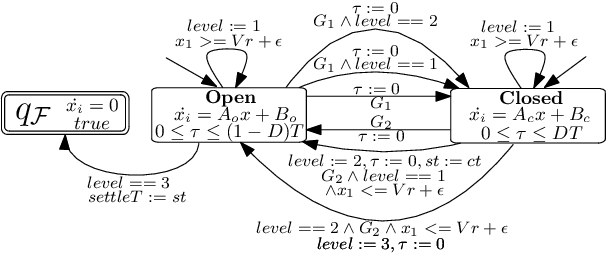}%{buckFA_settlingTime_ha.eps}
	\caption{LSHA for feature {\tt Settle Time}. $G_1$ and $G_2$ are transition guards between the locations as in Figure~\ref{fig:ha:buck}.}\label{fig:prductAutSettleTime}
	\end{figure}    
	
	The use of variable $level$ in $\mathcal{H}_{F}$ allows us to avoid the typical blow-up that results from the standard product construction used in~\cite{Costa16}. Given a HA with $N$ locations and a feature F with $K$ sub-expressions, ignoring pause locations, the product automaton in~\cite{Costa16} has $N\times(K+1)$ locations, while the one here has {\em atmost} $N\times (1+2\times(K-1))$ locations. Note that in~\cite{Costa16} first-match semantics isn't used. Also, the number of locations in the LSHA  here would reduce to $N+K+1$ if sub-expressions contain events over location labels, and further to $N+1$ if restrictions on monotonicity are lifted. We observe that, although the theoretical worst case bound for the product defined here is worse than that of a traditional product, in practice the use of leveling enables faster analysis of features. For instance without the variable $level$ the vanilla product from~\cite{Costa16} (where location copies are made for each level) when analyzed by SpaceEx, under equivalent analysis settings, takes 1m:30s and twice the memory with 7 locations as opposed to the 20s with 3 locations and the variable $level$ for the {\tt Settle Time} feature in Example~\ref{ex:settlingTime}.  %During a reachability analysis, this reduction greatly benefits time and memory utilization allowing reachability tools to better scale with the use of features. By flattening the automaton with the use of variable $level$, we avoid the quadratic increase in the number of locations, allowing reachability tools to better optimize time and memory from the reduction in the number of symbolic states being stored and analyzed, allowing reachability tools to yield results faster. For instance without the variable $level$ the vanilla product from~\cite{Costa16} (where location copies are made for each level) when analyzed by SpaceEx, under equivalent analysis settings, takes 1m:30s and twice the memory with 7 locations as opposed to the 20s with 3 locations and the variable $level$ for the {\tt Settle Time} feature in Example~\ref{ex:settlingTime}. 
	The LSHA for {\tt Settle Time} is shown in Figure~\ref{fig:prductAutSettleTime}.
	
	\subsection{Feature Evaluation}
	The product automaton is, by construction, a {HA}. We compute all values of $\mathcal{F}_{name}$ reachable in location $q_\mathcal{F}$ of $\mathcal{H}_F = M_F \bowtie \mathcal{H}$. Using reachset computation tools on $\mathcal{H}_F$, a projection of the entire reachset $\mathcal{R}$ in location $q_\mathcal{F}$ on $\mathcal{F}_{name}$ gives us an overapproximation of the reachable range of values for $\mathcal{F}_{name}$. A reachset computation tool may computes $\mathcal{R}$ in various ways. In this article we use tool SpaceEx. The feature range  $[\mathcal{F}_{min},\mathcal{F}_{max}]$ is computed as follows:
	
	\begin{mycenter}[0.1em]
	$\mathcal{F}_{min} = \min\limits_{\forall \sigma \in R}\eta(\sigma[\mathcal{F}_{name}])$;\\% \par
	$\mathcal{F}_{max} = \max\limits_{\forall \sigma \in R}\eta(\sigma[\mathcal{F}_{name}])$;	
	\end{mycenter}
	where $\eta(\sigma[\mathcal{F}_{name}])$ is the valuation of $\mathcal{F}_{name}$ in state $\sigma$.	
	
	%It is important to note that the reach set $R$ is symbolically computed and may be represented in terms of {flowpipes (convex hulls of reachable states) as polyhedral~\cite{AB01,dang04} or support functions~\cite{spaceex11}; a support function represents a convex set by attributing to each direction in $\mathbb{R}^n$, for $n$ directions, the signed distance of the furthest point in the set to the origin in that direction. $R$ may also be computed from a set of SMT clauses as done in~\cite{KongGCC15,GaoKC13}.} In the next section, we dwell on the extremal analysis of R.
	%-----------------------------------------------------------
	\section{Method-2: Feature Analysis of Corner Cases}\label{sec:SMT}
	Methodology-1 demonstrates how reachability analysis tools are used to compute estimates on the range of feature values. However, these tools do not show how extremal values can be reached. Additionally, due to over-approximation errors, corners of the feature interval generated by these tools may not be realizable. SMT solvers can generate reachability proofs for a reachable goal state. For a feature, this means that a proof of how feature value $\hat{f}$ is realized is constructed in terms of a concrete trace for which the evaluating the feature yields $\hat{f}$.
	SMT solvers for reals~\cite{GaoKC13} use decision procedures that use overapproximation techniques. The analysis is bounded in the values of SMT variables and in the number of discrete automaton transitions (a hop bound). Hence the outcome is an overapproximation of a bounded reachability question. If realistic and sufficiently large bounds are used, the boundedness is acceptable, since for most realistic systems the domains of variables in the {HA} model are bounded. To improve reliance on the results obtained, bounds used must come from knowledge of the design and results must be interpreted in terms of these bounds. Due to overapproximations, the proof of reachability for a goal in SMT could be fictitious, nevertheless it provides insight to build simulations to verify the reported scenario. In practice we find that overapproximations produced by reachset computation tools like SpaceEx are larger than that produced when using SMT solvers.
	
	It is important to note that application of Methodology 1 in Section~\ref{sec:featureTransformations} reduces the problem of feature analysis to a reachability question. The generality of the algorithm allows it to be used with a variety of reachability analysis tools.	SMT solvers, by nature, are less inclined to perform flow-pipe analysis (which generates an overapproximation of the state-space) and more inclined towards finding a {\em single} run that satisfies a goal constraint. This therefore becomes a challenge when we relate the notion of identifying the interval of values a feature can take for the {HA}, when using SMT solvers. We answer the following questions:
	\begin{enumerate}
	\item How would a reachability question for computing the feature interval be posed as an SMT solver goal?
	\item Any such goal will only yield a single feature value, and not a range. How would one then compute the extreme feature values?
	\end{enumerate}
	
	Once the range of feature values is identified using SMT, the SMT solver can  provide a satisfying trace, thereby solving the input selection problem for analyzing corner cases for features. 
	
	A summary of the methodology used to compute the feature range using SMT is shown in Figure~\ref{fig:featureSearch}.
	The feature to be analyzed is expressed in FIA. The {HA} model along with the feature is taken through the transformation step (Methodology-1). The tuned model is an implicit representation of all legal executions of the automaton, biased toward computing the feature attribute. 
	
	The SMT question about feature value $f$ is as follows:
	\begin{mycenter}[0.1em]
	{\em  \fontsize{9.9}{10} \selectfont Is there a run of the automaton that results in feature value f?}
	\end{mycenter}
	
	On the other hand, the Feature Range Analysis must answer the question:
	\begin{mycenter}[0.1em]
	{\em What is the range of feature attribute values for $\mathcal{H}$?}
	\end{mycenter}

	\begin{figure}[t!]
		\centering
		\includegraphics[width=8.5cm]{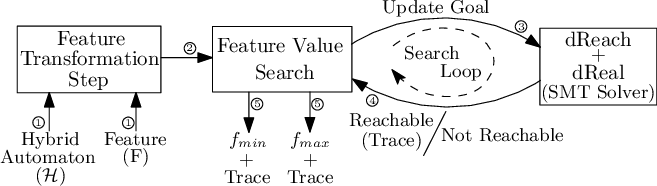}
		\caption{Outline of the Feature Value Search using SMT} \label{fig:featureSearch}
	\end{figure}
	
	To bridge this gap, we use a two part reduction described in Sections~\ref{sec:smt_modelling} and ~\ref{sec:smt_explore}.
	\subsection{SMT based modelling of the Hybrid Automaton Dynamics}\label{sec:smt_modelling}
	The {HA}, along with the various model constraints such as locations, their dynamics and invariants, transitions between locations and transition guards and resets, are modeled as clauses in SMT. We use the translator dReach~\cite{KongGCC15}, which in turn uses SMT solver dReal~\cite{GaoKC13} for modeling and analyzing hybrid behaviours over reals. dReal internally maintains the coupling between HA variables during its computation steps using these constraints. We now discuss the caveats of the decision outcomes presented by dReach/dReal. dReal solves the {\em $\delta$-decision problem}, to decide if a given formula is false or $\delta$-true (dually, whether it is true or $\delta$-false). An SMT formula is $\delta$-true if it remains true under $\delta$-bounded numerical perturbations to atomic clauses in the formula~\cite{GaoKC13}. For a feature, this means that a feature goal is reachable under $\delta$-bounded numerical perturbations to the goal, and sentences describing the system~\cite{KongGCC15}. Since realistic hybrid systems interact with the physical world, it is impossible to avoid slight perturbations. Hence, this is a very useful result as it gives feature values that are reachable under reasonable choices for $\delta$~\cite{GaoKC13}. The $\delta$-decision problem has been shown to be decidable for first order sentences over bounded reals with arbitrary Type 2 computable functions (real functions that can be approximated numerically, such a polynomials, trignometric functions, and Lipchitz-continuous ODEs). dReal guarantees the result of  unsatisfiability of a goal $G$ over $K$ transitions (hops) with $\delta$ perturbations on sentences describing the system.
	
	The model is {\em unrolled} in terms of number of transitions, upto the given bound $K$. For instance, if our goal were to reach location $q_\mathcal{F}$ in the location graph of Figure~\ref{fig:prductAutSettleTime}, a minimum of five transitions would be required, resulting from an unrolling of the model six times starting with location {\em Closed}, to reach location $q_\mathcal{F}$, reachable only when $level$ is 3, as shown below:
	
	\begin{mycenter}[0.1em]
	\includegraphics[scale=1]{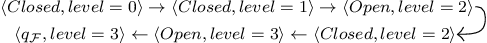}
	\end{mycenter}
	The encoding of a {HA} as SMT clauses for dReal can be found in Ref.~\cite{KongGCC15}. %{The clauses generated by the tool dReach are evaluated by the solver, dReal~\cite{GaoKC13}. dReal uses the theory of $\delta$-Computability\cite{GaoKC13} to decide if the goal is $\delta$-true or false.}  
	
	\subsection{Feature Range Exploration}\label{sec:smt_explore}
	To compute the extremal feature values and their corresponding traces, search techniques are used to explore the feasible set of feature values and progressively refine the corners of the feature range. SMT solvers take a {\em goal} statement, and a hop bound as input and respond indicating whether or not the goal is reachable. The feature range computed on a hop-bounded SMT-based search is therefore an under-approximation of what is obtained using an infinite trace length. An increase in the bounds can result in a larger feature range. A low value of the hop bound $K$ can yield a severely under-approximated feature range (ignoring feature values reachable via transition paths longer than $K$). In Section~\ref{sec:expts} we discuss a heuristic for choosing a value for $K$. Using an appropriately large $K$ ensures that the computed feature range safely over-approximates the reachable interval of feature values.
	
	To begin, no feature value $\mathcal{F}$ is initially known. If a behaviour contributing to the feature exists, its feature value will either be positive (including zero), or negative. Therefore, initially the search uses goals $\mathcal{F} < 0$ and $\mathcal{F}\geq 0$ as pivots about which to begin. The algorithm pushes a pivot as far as possible in each direction to find the corners of the feature range. However, since it is not known how {\em far} to push the pivot, a combination of exponential expansion and bisection is used to identify the corners of the feature range.
	
	\begin{figure}[t!]
	\centering
	\includegraphics[scale=0.9]{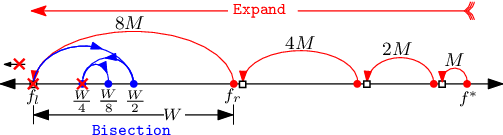}
	\caption{Using {\tt Expand-Bisect} to compute the left corner}\label{fig:ETS-moves}
	\end{figure}
	
	For pivot value {$f^*$}, a goal that trivially checks if $\mathcal{F} < f^*$ may return a feature value very close to $f^*$ resulting in repeatedly finding values of $f^*$ that are within close proximity of each other. Therefore the algorithm explores the feature space in steps of size $2^{i-1} \times M$, where $i$ is the $i^{th}$ expansion step. $M$ is a search parameter chosen by the designer. For instance, in the computation of the settling time for a buck regulator, $M = 10^{-6}$, because the feature value is in the order of $\mu$s. Expansion tries to find a feature value further than the current pivot $f^*$ by an amount $M$, and pushes the new value further by a recursive call to itself on the new pivot with a step size of $2\times M$. When no new pivot is found, a bounded interval search ensues to find a feature value within a distance of $M$ from the last known pivot $f^*$. The search refines the interval tapering it towards the corner of the feature range. Figure~\ref{fig:ETS-moves} depicts this strategy, starting at a pivot $f^*$ with boxes indicating steps during the expansion, red circles are new pivots, and red crosses indicate no feature value found. The bisection search terminates when the interval containing the feature corner has a width of $\epsilon$ (the error tolerance) or less.
	
	A bounded interval search for the left corner is now described. For the interval $W=[f_l,f_r]$, the midpoint $mid$ is used as a pivot. The algorithm looks for a feature value, $f^*$, in the leftmost interval $[f_l,mid]$. If found, it proceeds to search the interval $[f_l, f^*-\epsilon]$. The shrinking of the right interval boundary by $\epsilon$ is consistent with the precision requirements of the algorithm and excludes the already found feature value $f^*$ from the search. If no feature value is found in $[f_l,mid]$, then the search moves to $[mid,f_r]$. If a feature value is found in $[mid,f_r]$, then it forms the basis of a new pivot and the algorithm is recursively called on the interval $[mid, f^*-\epsilon]$. Every interval search uses the last feature value found as the pivot and excludes it from future searches. If no new feature value is reachable about the last pivot, then the last known reachable feature value and its corresponding trace (returned by the SMT solver as evidence of a satisfiable instance) are returned. The steps of exploration for the right corner mirror those for the left corner.
	
	Testing if a feature value $f^*$ is reachable involves invoking the SMT solver to answer the question specified as the goal $G$ in the context of the model $\mathcal{H}_F$. A `YES' answer of the solver returns a trace, ensuring that the dynamics, invariants and guard conditions of $\mathcal{H}_F$ are not violated. A `NO' answer, returns an {\tt EMPTY} indicating that no behaviour of $\mathcal{H}_F$ yields a feature value specified in G.
	
	%The We explain how the {\tt ETS} algorithm searches for the left corner of the feature interval{, while the search or the right corner mirrors choices of the search for the left corner.} {\tt ETS} {is} inspired by bisection search algorithms. 
	
	If an interval for the feature is known, via a primary reachability analysis using a coarse overapproximation in a tool like SpaceEx, then the bounded bisection algorithm can be applied to search for the feature corners within the known estimate of the feature range.
	
	%-----------------------------------------------------------
	
	\section{Case Studies and Experimental Results}\label{sec:expts}
	This section discusses the various models we have used in our analysis, and data from our results, comparing the feature analysis of {HA} using the reachability analysis tool SpaceEx versus using the SMT tool dReach.
	
	We compare the results of analyzing the following models and features:
	\begin{enumerate}
	\item \textbf{Battery Charger~\cite{Costa16}:} Time for the battery to charge to its rated voltage; Time for the battery to restore charge while in its maintenance mode.
	\item \textbf{Buck Regulator~\cite{buck}:} Time for the output of the Buck regulator to settle; Peak voltage overshoot of the voltage response curve for the regulator.
	\item \textbf{Nuclear Reactor Temperature Control~\cite{alur96-2}:} Unsafe Operating Temperature of the reactor.
	\item \textbf{Adaptive Cruise Control~\cite{hybridBenchmarks}:} Time to capture cruise speed from a specific velocity; Time to capture cruise speed while in any velocity within a range of velocities.
	\end{enumerate}
	
	% We provide detailed descriptions of the hybrid system models and the formal descriptions of features used in an extended version of this article available at {\tt arXiv:1711.00669 [cs.LO]} .
	
	\begin{figure}[t!]
	\centering
	\includegraphics[scale=0.8]{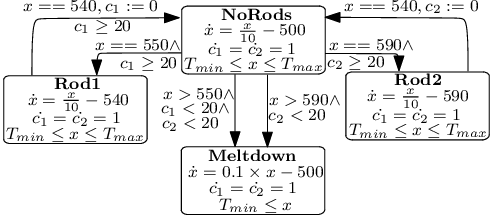}
	\caption{Temperature Control of an Atomic Reactor}\label{fig:nuclear}
	\end{figure}
	
	Here, we pay special attention to the condition of meltdown for an atomic reactor cooling strategy. The temperature control strategy for a nuclear reactor~\cite{alur96-2} is designed to insert a cooling rod into the reactor with the aim of maintaining the temperature of the reactor below the threshold for meltdown and above the threshold for sustaining the nuclear reaction. Mechanical constraints prevent both rods from being inserted simultaneously, and requires each rod to be given a resting period of 20 time units before re-insertion. An adaptation of the HA of~\cite{alur96-2} is shown in Figure~\ref{fig:nuclear}. The feature for analyzing the condition of meltdown is expressed as follows:
	
	\begin{example}\label{ex:unsafe}
	\emph{\textit{Unsafe Operating Temperature:} Reactor temperatures that if reached can lead to reactor meltdown.}
	\end{example}
	{\scriptsize
		\begin{verbatim}
		feature unsafe();
		begin
		  var temp;
		  (c1<=20 && c2<=20 && x>=550), temp = x 
		  ##[0:$] (c2<=20 && x>=590) |-> unsafe = temp;
		end	\end{verbatim}
}

The condition of meltdown occurs when the reactor temperature, $x$, rises above the safe threshold, $T_{safe}$. The reactor can be in a state in which $x < T_{safe}$, but has crossed a point-of-no-return, i.e. neither control rod can be inserted, inevitably leading to a state of reactor meltdown. A safety property that checks for the safe operation of the reactor, with a traditional model checking approach, would only yield one of the many possible failures. However, it is of greater interest to identify the minimum temperature at which such a failure can occur. Knowledge of this corner enables a designer to design a suitable strategy for managing the rods. A feature analysis, unlike traditional model checking, yields these corner cases. Furthermore, for such a feature where only boundary events characterizing a failure are known, the second technology proposed (feature analysis with SMT) can provide the precise event sequence that yields a feature match, thereby filling the gap specified in the {\tt \#\#[0:\$]} construct. Figure~\ref{fig:unsafe} shows a corner case obtained using the methodology of Section~\ref{sec:SMT}, in which the red vertical line marks a point-of-no-return. Observe that from this point {\tt x} rises, passing through {\em safe} temperatures and beyond into the unsafe region of meltdown. In this scenario, both rod-timers are below their thresholds, preventing their insertion.

\begin{figure}[t]
	\centering
	%\hspace*{-1em}
	\includegraphics[width=0.48\textwidth,height=4.8cm]{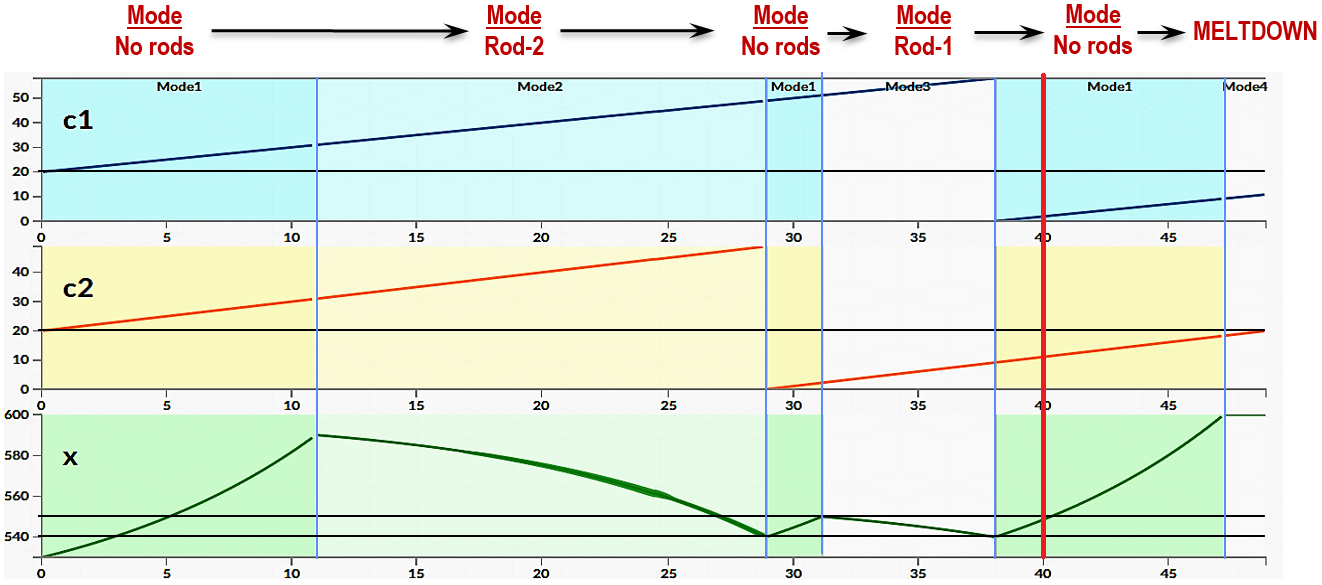}
	\caption{Unsafe Operation of a dual rod temperature control in an atomic reactor: an extreme value trace.}\label{fig:unsafe}
\end{figure}

Note that in Examples 1 and 2, both features used a single local variable, but were  able to express very interesting behaviours. However, in general more local variables may be required to express complex quantities over more intricate behaviours~\cite{dyFET}.

A feature based formal analysis was performed on the models and features outlined, the results of which are described in Table~\ref{table:featureRanges}. Table~\ref{table:Overshoot} compares the results of using various SpaceEx analysis parameters for analyzing the {\tt Overshoot} feature of a 7V Buck Regulator. We demonstrate the analysis of both strategies for feature analysis on four systems that cover both the AMS domain and the control domain. Both strategies have been implemented in a unified tool-flow. The tool is run on an Intel(R) Core(TM) 2 Duo CPU T6400 having two cores, each running at 2.00Ghz with 4GB of DDR2 RAM. For each system, the HA model and the features described are inputs to the tool. The tool then computes the LSHA for feature analysis. The feature range is then computed using the reachability tool SpaceEx (with the STC scenario and 8 template directions, octagons), and the SMT-based search using Methodology-2. 
%\hspace*{-1em}
\begin{table}[tb!]
	\begin{mycenter}[0.1em]
		\centering
		
		\scriptsize{
			{	\setlength\tabcolsep{1.5pt}
				\hspace*{-1em}\begin{tabular}{|c|c|c|c|c|c|c|}
					\hline
					{\bf Feature}  &\multicolumn{2}{c|}{\bf Size of Set}   & {\bf Algorithm} & {\bf CPU-Time}        &\multicolumn{2}{c|}{\bf Feature Range}    \\ \cline{2-3}\cline{5-7}
					{\bf Name   }  &{\bf $Q_F$}        & {\bf $X_F$}       &                 &{\bf (mins : secs)}	 &      {\bf Min} &        {\bf Max}	    \\ \hline
					%\cline{2-3}
					%\hline

					\multicolumn{7}{|c|}{\cellcolor{lightgray}{\textbf{Test Case: Battery Charger}}} \\
					\hline
					
					Charge  &  9 & 7 & {\bf Reach} & 0m : 10s & 1hr 49min  &  3hr 15min \\ \cline{4-7}
					Time    &    &   & {\bf SMT} & 0m: 20s & 2hr 15min   & 2hr 31min \\ \hline
					
					Restoration &  9 & 7  & {\bf Reach} & 0m : 19s &  10in 12sec & 48min 58sec \\ \cline{4-7}
					Time        &    &    & {\bf SMT} & 0m : 15s & 16min 40sec & 18min 30sec \\ \hline 
					
					\multicolumn{7}{|c|}{\cellcolor{lightgray}{\bf Test Case: Cruise Control Model}} \\
					\hline
					Speed Capture  & 10 & 8 & {\bf Reach} & 0m : 18.4s & 37.23 sec & 48.43sec      \\ \cline{4-7}
					Precise (k=40) &  &  & {\bf SMT} & 1m : 15.2s & 41.18 sec & 43.68sec    \\ 
					\hline
					
					Speed Capture         &  10  &  8 & {\bf Reach} & 0m : 38.92s & 33sec & 48.43sec      \\ \cline{4-7}	
					Range, (k1=20, k2=40) &   &  & {\bf SMT} & 0m : 24.56s & 37.3sec  & 40.4sec    \\ \hline
					
					\multicolumn{7}{|c|}{\cellcolor{lightgray}{\bf Test Case: Nuclear Reactor Control}} \\
					\hline
					Unsafe Operation  &	8 & 6 & {\bf Reach} & 0m : 7s & 549.9$^\circ$ & 599.9$^\circ$ \\ \cline{4-7}
					Temperature	      &   &   & {\bf SMT} & 0m : 52s & 550$^\circ$ & 600$^\circ$ \\  
					\hline

					%Safe Operation & 6  & 6 & {\bf Reach} & 0m : 21s & 539.9$^\circ$ & 590.12$^\circ$ \\ \cline{4-7}	
					%Temperature    &    &   & {\bf SMT} & 0m : 14s & 540$^\circ$ & 589.19$^\circ$  \\  \hline \hline 
					\multicolumn{7}{|c|}{\cellcolor{lightgray}{\bf Test Case: 5V Buck Regulator}} \\

					\hline
					
					Settle Time & 4 & 7 & {\bf Reach} & 0m : 16s & 94.17$\mu$s  & 124.167 $\mu$s\\ \cline{4-7}
					&   &   & {\bf SMT} & \multicolumn{3}{c|}{Out of Memory}\\ %\cline{4-7}
					\hline
					
					Overshoot&  4  & 7  & {\bf Reach}  & 0m : 03s & 5V  & 6.138V \\\cline{4-7} 
					&     &    & {\bf SMT} &  69m : 45s  & 5V  & 5.14V \\ 
					\hline

				\end{tabular}
				\caption{Results for Formal Feature analysis}	\label{table:featureRanges}
			} 
		}	
	\end{mycenter}

	\begin{mycenter}[0.1em]
		\centering
		
		\scriptsize{
			{	\setlength\tabcolsep{1.5pt}
				\hspace*{-1em}\begin{tabular}{|c|c|c|c|c|c|c|}
					\hline
					{\bf No. of Template}  &{\bf Flow-pipe}  & {\bf SpaceEx}     & {\bf Iterations} & {\bf CPU-Time}     &\multicolumn{2}{c|}{\bf Feature Range}\\ \cline{6-7}
					{\bf Directions }      &{\bf Tolerance}  & {\bf Algorithm}   &{\bf Taken}       &{\bf (mins : secs)} &     {\bf Min} &        {\bf Max}	    \\ \hline
					
						& 1 	&{\cellcolor{lightgray} STC} & {\cellcolor{lightgray}37} & {\cellcolor{lightgray}0.216999} & {\cellcolor{lightgray}1.76V} & {\cellcolor{lightgray}12.65V}\\\cline{3-7}
						&   	& LGG & 11 & 0.234999 & 4.41V & 9.96V \\\cline{2-7}
					
					4	& 0.1 	& {\cellcolor{lightgray}STC }& {\cellcolor{lightgray}36 }& {\cellcolor{lightgray}0.212999} & {\cellcolor{lightgray}5.14V} & {\cellcolor{lightgray}12.11V}\\\cline{3-7}
						&     	& LGG & 55 & 0.620999 & 6.63V & 9.90V \\\cline{2-7}
						
						& 0.01  & {\cellcolor{lightgray}STC} & {\cellcolor{lightgray}60 }& {\cellcolor{lightgray}0.517999} & {\cellcolor{lightgray}6.78V} & {\cellcolor{lightgray}9.12V}\\\cline{3-7}
						&   	& LGG & 57 & 0.814999 & 6.91V & 9.0V \\\hline
						
						& 1 	& {\cellcolor{lightgray}STC} & {\cellcolor{lightgray}85} & {\cellcolor{lightgray}4.52299} & {\cellcolor{lightgray}2.00V} & {\cellcolor{lightgray}12.04V}\\\cline{3-7}
						&   	& LGG & 11 & 1.51699 & 4.41V & 9.32V\\\cline{2-7}
						
					8	& 0.1 	& {\cellcolor{lightgray}STC} & {\cellcolor{lightgray}28} & {\cellcolor{lightgray}0.65899} & {\cellcolor{lightgray}5.17V} & {\cellcolor{lightgray}12.06V}\\\cline{3-7}
						&     	& LGG & 55 & 2.74399 & 6.63V & 9.19V \\\cline{2-7}
						
						& 0.01  & {\cellcolor{lightgray}STC} & {\cellcolor{lightgray}72} & {\cellcolor{lightgray}3.76499} & {\cellcolor{lightgray}6.82V} & {\cellcolor{lightgray}9.08V}\\\cline{3-7}
						&   	& LGG & 57 & 4.77099 & 6.92V & 9.0V \\\hline
				\end{tabular}
				\caption{SpaceEx analysis of {\tt Overshoot}: 7V Buck Regulator}	\label{table:Overshoot}
			} 
		}	
	\end{mycenter}
\end{table}

In Table~\ref{table:featureRanges}, the size of the transformed automata in terms of the number of locations (in set $Q_F$) and number of variables (in set $X_F$) are also shown. The column titled {\em Algorithm} indicates how the feature range was computed, with "Reach" indicating the use of Methodology-1 with SpaceEx as the compute engine, and "SMT" indicating the use of Methodology-2 to refine the range computed using the former. For each feature both corners of the feature range are reported along with the time taken to compute the range for each methodology. Both SpaceEx and the SMT analysis are bounded by introducing a global clock with an upperbound on time. The global clock is a variable that is part of the state of the LSHA. In SpaceEx, with the STC scenario, given the way SpaceEx does fix-point computation this variable must be bounded to achieve termination. Note that the need for the bound comes from the tool used and is not a limitation of this methodology. For the SMT analysis, in addition to a bound on the variables, we use a transition hop bound ($K$) of 15 transitions for all test cases except for the Buck Regulator, for which a bound of 50 transitions was used. In practice $K$ is incrementally increased until we are satisfied with the result. We use knowledge of the diameter of the LSHA graph and the number of subexpressions in the feature sequence expression to decide on a value for $K$. We reiterate that $K$ is a bound on the discrete transitions of the HA. Within a location a large number of clauses may be generated for the evolution of the HAs continuous variables. It is important to note that the SpaceEx tool computes the feature range in one sweep of the reach set; however, multiple iterations of the SMT tool (between 15 to 20 in our experiments) are involved in computing the feature range using "SMT". Additionally for the feature computing the {\em Unsafe Operation Temperatue}, the feature ranges produced by SpaceEx and the SMT tool show errors of 0.1. We attribute this to the precision of representation for floating-point numbers used by the tools.

Note that for the feature "Settle Time" of the Buck Regulator, the methodology using SMT exceeds memory bounds on our systems. We attribute this to the fact that the Buck Regulator frequently switches between locations of the automaton, with more than 50 transitions made within a very short span of time (time from the perspective of the Buck Regulator). The solver takes an inordinate amount of time to compute this. Due to the large number of transitions taken, the number of SMT clauses generated becomes too large for the solver to handle and leads to the solver running out of memory. We conclude that for systems having a high switching frequency, SpaceEx can be used with a resolution smaller that $10^{-6}$, for which it takes in the order of a few seconds to a few minutes to compute the feature range (depending on the chosen resolution).

For the models used here, it is shown that the feature range produced by the SMT solver is typically tighter than that obtained using SpaceEx. Both methodologies were employed using similar error tolerances. Note that the methodology using the SMT solver requires more CPU resources as indicated by a higher value in the column for {\em CPU-Time}. {The feature transformation methodology itself scales well with reachability tools and SMT. The time for analysis is dependent on the tools used. The tool SpaceEx has been used extensively for the analysis of {HA} and scales well for the models on which we have demonstrated the feature analysis approach. SpaceEx, in benchmarks has shown to be capable of handling systems with more than 100 variables~\cite{spaceex11}. The time and memory to compute a feature range grows exponentially with an increase in the hop bound when using SMT, and is attributed to a growth in the number of SMT clauses for larger hop bounds.} 

%-----------------------------------------------------------	
\section{Conclusion}\label{sec:conclusions}
Features help capture the designer's intent to quantify how the system behaves. A feature defines a real-valued evaluation function over a specific set of traces. By design they are more flexible than assertions (such as in STL) for specifying quantitative measures, at the cost of being more rigid in their expression of sets of traces (restricted to sequences of predicates and events). 

% . In past work, methodologies for the feature analysis of mixed-signal circuits through simulation was demonstrated and a first attempt at the formal feature analysis of hybrid automata was studied. 

%In the control domain, translation of automata based models into control tasks mapped to embedded platforms are becoming increasingly common. For example, automatic code generation from Simulink/Stateflow models is supported in MATLAB. In the circuit domain {HA} models can be translated into behavioural models in languages such as VerilogA, VerilogAMS and SPICE. {Additionally, in~\cite{Costa_VLSID_2017}} feature analysis has been used to tune {HA} models through user feedback and automatically generate feature-accurate VerilogAMS models. The key aspect of these strategies is in generating a wrapper interface around the model to allow it to be seamlessly integrated into a simulation flow with other modules.

This article aims to assist designers in generating better designs, by automating the task of feature analysis and providing useful feedback of corner case behaviours using SMT. Features are automatically transformed into feature automata that are composed with the model and the composition is analyzed by off-the-shelf reachability solvers. The improved first-match semantics employed for features in this article more directly reflects the intent of designers and is incorporated into a new product automaton construction. Although the worst case bounds for the proposed product construction (using leveling) are worse than those of a more traditional product, in practice we see a 4x speedup and half the memory utilization during feature analysis with reachability solvers. We also provide an algorithm for computing ranges of feature values that uses SMT, which in practice produces tighter feature ranges. In some cases, typically associated with models having fewer location switches, the SMT-based algorithm also yields results faster or in time comparable to SpaceEx. The present work assumes piecewise monotonic and piecewise affine dynamics to accommodate existing tools. These assumptions can be lifted as tools mature to support urgent semantics and more complex dynamics. Efforts for such extensions~\cite{MinopoliF16a} are underway.

%-----------------------------------------------------------

\bibliographystyle{ACM-Reference-Format}
\bibliography{report}

%-----------------------------------------------------------	

\end{document}